\documentclass[11pt]{article}
\usepackage[top=1in, left=0.7in, right=0.7in,bottom=1in, paperwidth=8.5in, paperheight=11in]{geometry}
\usepackage[margin=10pt,font=footnotesize,labelfont=bf]{caption}
\usepackage{amsmath}
\usepackage{graphicx}
\usepackage{wrapfig}
\usepackage{algorithm2e}
\usepackage[noend]{algpseudocode}
\usepackage{amsthm}
\usepackage{amssymb}
\usepackage{mathtools}
\usepackage{authblk}
\usepackage{comment}
\usepackage{color}
\usepackage{soul}
\usepackage{textcomp}

\newcommand{\eps}{\varepsilon}
\newcounter {lemcounter}
\newtheorem{theorem}{Theorem}
\newtheorem{lemma}[lemcounter]{Lemma}

\newtheorem{observation}{Observation}

\newtheorem{cor}{Corollary}
\newcommand {\quality}[2]{Q^#2_{#1_1#1_2}} 
\newcommand {\Tquality}[2]{Q^#2(#1_1,#1_2)} 
\newcommand {\sPoint}{\mathfrak{s}}
\newcommand {\tPoint}{\mathfrak{t}}
\newcommand {\ShortMP}{{\sf wShortMP}}
\newcommand {\LongMP}{{\sf wLongMP}}
\newcommand {\MinEx}{{\sf MinEx}}
\newcommand {\MaxIn}{{\sf MaxIn}}
\newcommand {\acmq}{{\sf{ACM}}$\mathbb{Q}$\xspace}
\newcommand*\samethanks[1][\value{footnote}]{\footnotemark[#1]}

\makeatletter

\begin{document}
\title{Similarity of Polygonal Curves in the Presence of Outliers}
\author[1]{Jean-Lou De Carufel\thanks{Research supported by Fonds de Recherche du Qu\'{e}bec - Nature et Technologies}}
\author[1]{Amin Gheibi\thanks{Research supported by High Performance Computing Virtual Laboratory and SUN Microsystems of Canada}\thanks{Research supported by Natural Sciences and Engineering Research Council of Canada}}
\author[1]{Anil Maheshwari\samethanks}
\author[1]{J\"{o}rg-R\"{u}diger Sack\samethanks[2]\samethanks[3]}
\author[2]{Christian Scheffer}
\affil[1]{School of Computer Science, Carleton University, Ottawa, ON, Canada
  \texttt{\small jdecaruf@cg.scs.carleton.ca\\agheibi@scs.carleton.ca\\
anil@scs.carleton.ca\\sack@scs.carleton.ca}}

\affil[2]{Department of Computer Science, Westf\"{a}lische Wilhelms-Universit\"{a}t M\"{u}nster, Germany
  \texttt{\small Christian.Scheffer@uni-muenster.de}}

\maketitle
\thispagestyle{empty}
\abstract{
The Fr\'{e}chet distance is a well studied and commonly used measure to capture the similarity of polygonal curves. Unfortunately, it exhibits a high sensitivity to the presence of outliers. Since the presence of outliers is a frequently occurring phenomenon  in practice, a robust variant of Fr\'{e}chet distance is required which absorbs outliers. We study such a variant here. In this modified variant, our objective is to minimize the length of subcurves of two polygonal curves that need to be ignored (\MinEx\ problem), or alternately, maximize the length of subcurves that are preserved (\MaxIn\ problem), to achieve a given  Fr\'{e}chet distance.
An exact solution to one problem would imply an exact solution to the other problem. However, we show that these problems 
are not solvable by radicals over $\mathbb{Q}$ and that the degree of the polynomial equations involved is unbounded in general. 
This motivates the search for approximate solutions. 
We present an algorithm, which approximates, for a given input parameter $\delta$, optimal solutions for the \MinEx\ and \MaxIn\ problems up to an additive approximation error $\delta$ times the length of the input curves. The resulting running time is upper bounded by $\mathcal{O} \left( \frac{n^3}{\delta} \log \left( \frac{n}{\delta} \right)\right)$, where $n$ is the complexity of the input polygonal curves.\\ 
{\bf Keywords:} Fr\'{e}chet distance, similarity of polygonal curves, approximation, weighted shortest path.

  
  
\eject
\setcounter{page}{1}
\section{Introduction}\label{sec:introduction}
Measuring similarity between two polygonal curves in Euclidean space is  a well studied problem in computational geometry, both in practical and theoretical setting. It is of practical relevance in areas such as pattern analysis,  
shape matching and clustering. It is of theoretical interest as well since the problems in this domain are 
fairly challenging and lead to innovative tools and techniques.
The Fr\'{e}chet distance - one of the widely used measures for similarity between curves - is intuitive and takes into account global features of the curves instead of local ones, such as their vertices \cite{Alt95,Buchin11,Driemel12}.
Despite being a high quality similarity measure for polygonal curves, it is very sensitive to the presence of outliers. 
Consequently, researches have been carried out to formalize the notion of similarity among a set of polygonal curves that tolerate outliers. 
They are based on intersection of curves in local neighborhood \cite{Kreveld11}, topological features \cite{Buchin10}, or  adding flexibility to incorporate the existence of outliers \cite{Driemel12}. In \cite{Driemel12}, Driemel and Har-Peled discuss a new notion of robust Fr\'{e}chet distance, where they allow $k$ shortcuts between vertices of one of the two curves, where $k$ is a constant specified as an input parameter. 
They provide a constant factor approximation algorithm for finding the minimum Fr\'{e}chet distance among all possible $k$-shortcuts. 
One drawback of their approach is that a shortcut is selected without considering the length of the ignored part. Consequently, such shortcuts may remove a significant portion of a curve. As a result, substantial information about the similarity of the original curves could be ignored. A second drawback of their approach is that the shortcuts are only allowed to one of the curves. Since noise could be present in both curves, shortcuts may be required on both to achieve a good result. For example, Figure \ref{fig:stretchedP}a shows two polygonal curves that both need simultaneously shortcuts to become similar.

In this paper we discuss an alternative Fr\'{e}chet distance measure to tolerate outliers;  
it incorporates the length of the curves and allows the possibility of shortcuts on one or both curves. 
We consider two natural dual perspectives  of this problem. They are outlined as follows using the common dog-leash metaphor for Fr\'{e}chet distance.\\
\noindent {\bf Min-Exclusion (\MinEx) Problem}: If a person wants to walk on one curve and his/her dog on the other one, for a given leash length $\eps\ge 0$, we wish to determine a walk that minimizes the total length of all parts of the curves that need a leash length bigger than $\eps$.\\
\noindent {\bf Max-Inclusion (\MaxIn) Problem}: We are looking for a walk that maximizes the total length of all parts of the curves that a leash length less than or equal to $\eps$ is sufficient.

Observe that the solution for one problem leads to a solution for the other problem.
An exact solution for these problems is presented in \cite{Yusu09}, where the distances are measured using (more restrictive and much simpler) $L_1$ and $L_\infty$ metrics. 
In Section \ref{sec:unsolv}, using Galois theory,  we  show that these problems are not solvable by radicals over $\mathbb{Q}$, when 
distances are measured using $L_2$-metric. (It is natural to study Fr\'{e}chet distance problems in $L_2$-metric, see e.g. \cite{Alt95}.) This suggests that we should look for approximation algorithms.
A $(1-\delta)$-approximation algorithm for the \MaxIn\ problem had been outlined in \cite{Yusu06}, where $\delta$ is the approximation factor. As we show in Section \ref{sec:FPTAS} this analysis is incorrect (see also \cite{Yusu13}). Therefore, to the best of our knowledge, no FPTAS (Fully Polynomial-Time Approximation Scheme) exists for this problem.
In this paper,  we provide algorithms that approximate solutions for the \MinEx\ and \MaxIn\ problems up to an additive approximation error $\delta$ times the length of the input curves.

 \begin{figure}[htb]
    		\begin{center}
     			 \begin{tabular}{cc}
       				\includegraphics[height=3.63cm]{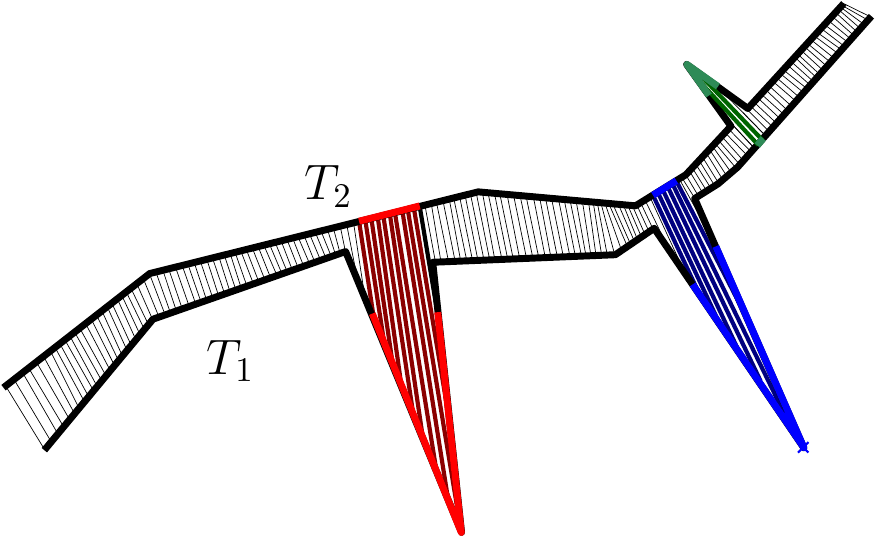} & \includegraphics[height=3.63cm]{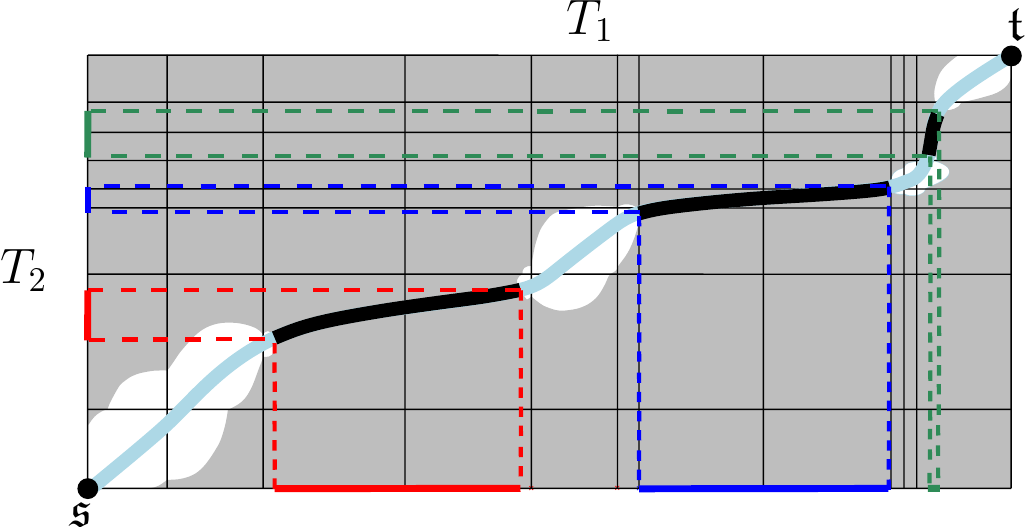}\\
       				{\small a) Two polygonal curves $T_1$ and $T_2$.} & {\small b) Free-space diagram for $T_1$ and $T_2$.}\\
      			\end{tabular}
    		\end{center}
    		\vspace*{-12pt}
    		\caption{a) A possible solution is illustrated by the connecting lines between the parameterizations for $T_1$ and $T_2$. The subcurves on both polygonal curves that should be ignored are illustrated by the blue, red and green subcurves on $T_1$ and $T_2$. So, $\Tquality{T}{B}$ is the summation of the lengths of the colored subcurves and $\Tquality{T}{W}$ is that of the black subcurves. b) The solution corresponds to an $xy$-monotone path in the deformed free-space diagram $F$. In this space, $\Tquality{T}{B}$ can be measured by  summing the lengths of its subpaths going through the forbidden space (shaded gray area), measured in the $L_1$-metric (similarly for $\Tquality{T}{W}$).}
  		\label{fig:stretchedP}
  	\end{figure}
\subsection{Preliminaries}\label{subsec:Preliminaries}
	Let $T_1:\left[ 0,1 \right] \rightarrow \mathbb{R}^2$ and $T_2: \left[ 0,1 \right] \rightarrow \mathbb{R}^2$ be two polygonal curves.  Their \emph{Fr\'{e}chet distance} is defined as the minimum \emph{leash length} required to walk only forwardly, in parallel on both $T_1$ and $T_2$, from the starting points to the ending points, at which the two walks could have different variating speeds. More formally, two monotone parameterizations $\alpha_1, \alpha_2: \left[ 0,1 \right] \rightarrow \left[0,1\right]$ define, for each time $t \in \left[ 0,1 \right]$, a matching $\left( T_1 \left( \alpha_1 \left( t \right) \right), T_2 \left( \alpha_2 \left( t \right) \right) \right)$ of one point on $T_1$ to exactly one point on $T_2$ and vice-versa. The needed leash length for the two parameterizations is defined as the maximum Euclidean distance of two matched points, over all times. Then, Fr\'{e}chet distance $\delta_F \left( T_1,T_2 \right)$ is defined as the infimum of the required leash lengths over all possible pairs of monotone parameterizations \cite{Alt95}: 
		\begin{equation}\label{eqn-parameter}
			\delta_F \left( T_1,T_2 \right) := \inf_{\alpha_1,\alpha_2: \left[ 0,1\right] \rightarrow \left[ 0,1\right]}  \max_{t \in \left[ 0,1 \right]} \left\lbrace |T_1 \left( \alpha_1 \left( t \right) \right) T_2 \left( \alpha_2 \left( t \right) \right)| \right\rbrace,
		\end{equation}
where $|.|$ denotes the Euclidean distance. For simplification, we say from now, that all considered parameterizations are monotone.
The corresponding \emph{Fr\'{e}chet distance decision problem} asks if there exist two parameterizations for a given leash length $\varepsilon$, realizing a Fr\'{e}chet distance between $T_1$ and $T_2$ that is upper bounded by $\varepsilon$.
In other words, it asks if it is possible to walk your dog with a given leash of length $\varepsilon$, such that you and your dog stay on your own curves.
For a fixed leash length $\varepsilon$, a pair of arbitrary points on the curves of $T_1$ and $T_2$ is called \emph{forbidden}, if their Euclidean distance is bigger than $\varepsilon$ and otherwise, it is called \emph{free}. We refer to a pair of parameterizations $\left(\alpha_1, \alpha_2\right)$ for $T_1$ and $T_2$ w.r.t. $\varepsilon$, as a possible \emph{solution} for $T_1$ and $T_2$. Analogously to the matching of points on the curves of $T_1$ and $T_2$, we define a pair of parameters $p_1$ and $p_2$ as \emph{forbidden} (\emph{free}), if the Euclidean distance between $T_1 \left( p_1 \right)$ and $T_2 \left( p_2 \right)$ is greater (equal to or less) than $\varepsilon$. In this context, the Fr\'{e}chet distance decision problem asks, if there exists a parametrization, such that all corresponding matchings are free \cite{Alt95}.

To decide, whether the Fr\'{e}chet distance between two polygonal curves is upper bounded by a given $\varepsilon$, the \emph{free-space diagram} is computed. 
The free-space diagram is a decomposition of  the parameter space $\left[0,1 \right] \times \left[0,1 \right]$ of $T_1$ and $T_2$ into two sets.
The first one is the \emph{forbidden-space}, which is defined as the union of all forbidden parameter pairs. 
The second set is the \emph{free-space}, defined as the complement of the forbidden-space. Let $n_1$ (respectively, $n_2$) be the number of segments of $T_1$ (respectively, $T_2$) and let $n=n_1+n_2$.
Since we consider the worst case running time, we assume, w.l.o.g., $n = n_1 = n_2$. The free-space diagram is a rectangle, partitioned into $n$ columns and $n$ rows. It consists of $n^2$ \emph{parameter cells} $C^{i,j}$, for $i,j = 1,...,n$, whose interiors do not intersect with each other. 
For each parameter cell $C^{i,j}$, there exists an ellipse such that the intersection of the area bounded by this ellipse with $C^{i,j}$ is equal to the free-space region of that cell \cite{Alt95}.  
Free-space of $C^{i,j}$ is denoted by $C_W^{i,j}$ and its forbidden space by $C_B^{i,j}$. We consider the free-space diagram in the Cartesian plane and assume that it lies axes aligned. We say that a point lies to the left (right) of another point, if its $x$-coordinate is smaller (greater) than the second one. Analogously, we say, that a point lies below (above) another point, if its $y$-coordinate is smaller (greater) than the second one. We define a point $s$ {\em dominates} another point $s'$, if and only if $s'$ does not lie to the right or above $s$.

For simplification, we say that paths, curves, edges, etc. are $xy$-monotone if they are nondecreasing in both $x$- and $y$-coordinates. We know that any pair of parameterizations $(\alpha_1,\alpha_2)$  corresponds to an $xy$-monotone path in free-space diagram, $\pi_{\sPoint \tPoint}$, connecting the bottom left corner of the diagram, $\sPoint$, to the upper right corner, $\tPoint$. So, deciding if the Fr\'{e}chet distance of $T_1$ and $T_2$ is upper bounded by $\varepsilon$, is equivalent to the decision: Does there exists an $xy$-monotone path in the free-space diagram for $\varepsilon$, connecting $\sPoint$ to $\tPoint$ and avoiding the forbidden space \cite{Alt95}? 
\subsection{Problem Definition}\label{subsec:ProblemDefinition}
Let $T_1$ and $T_2$ be two polygonal curves, each consisting of at most $n$ line segments; and $\varepsilon\ge 0$ be a constant.
%
For the \MinEx\ (respectively, \MaxIn) problem, the {\em quality} of a solution is the sum of the lengths of the subcurves of $T_1$ and $T_2$ that lie in the  forbidden (respectively, free) space  (see Figure \ref{fig:stretchedP} for an example).
Formally, for a given pair of parameterizations $\left( \alpha_1,\alpha_2 \right)$, let $\mathcal{B}_{\alpha_1\alpha_2} \subseteq \left[ 0,1 \right]$ be the closure of the set of times such that the corresponding parameter pairs are forbidden, and $\mathcal{W}_{\alpha_1\alpha_2} \subseteq \left[ 0,1 \right]$ be the closure of the set of times such that the corresponding parameter pairs are free. We define the quality of a solution $\left( \alpha_1, \alpha_2 \right)$ for the \MinEx\ problem $\quality{\alpha}{B}$ and for the \MaxIn\ problem $\quality{\alpha}{W}$ as follows (see Figure \ref{fig:stretchedP}a):

\begin{equation}
	\label{eq:quality}
		\begin{array}{lcr}
			\quality{\alpha}{B}& :=& \displaystyle\int_{t \in \mathcal{B}_{\alpha_1\alpha_2}} ||T_1 \left( \alpha_1 \left( t \right) \right)'|| dt + \int_{t \in \mathcal{B}_{\alpha_1\alpha_2}} ||T_2 \left( \alpha_2 \left( t \right) \right)'|| dt \\
			\quality{\alpha}{W} &:= &\displaystyle\int_{t \in \mathcal{W}_{\alpha_1\alpha_2}} ||T_1 \left( \alpha_1 \left( t \right) \right)'|| dt + \int_{t \in \mathcal{W}_{\alpha_1\alpha_2}} ||T_2 \left( \alpha_2 \left( t \right) \right)'|| dt
		\end{array}
	\end{equation}
where $||v||$ is $L_2$ norm of a vector $v$. We call $\left( \alpha_1, \alpha_2 \right)$ \emph{optimal} if it minimizes (respectively, maximizes) $\quality{\alpha}{B}$ (respectively, $\quality{\alpha}{W}$) and define the \emph{quality of $T_1$ and $T_2$ w.r.t. $\varepsilon$} as its value, i.e.,
	\begin{equation}
		\begin{array}{lcr}
			\Tquality{T}{B}&:=& \displaystyle\inf_{\alpha_1,\alpha_2:\left[ 0,1 \right] \rightarrow \left[ 0,1 \right]} \quality{\alpha}{B}\\  \Tquality{T}{W}&:=& \displaystyle\sup_{\alpha_1,\alpha_2:\left[ 0,1 \right] \rightarrow \left[ 0,1 \right]} \quality{\alpha}{W}
		\end{array}
	\label{eq:quality2}
	\end{equation}
	
This means that the quality of $T_1$ and $T_2$ is the minimum (respectively, maximum) sum of lengths of curves on $T_1$ and $T_2$ to be ignored (respectively, matched), to obtain a Fr\'{e}chet distance not greater than $\varepsilon$.
For a given $\varepsilon$, we would like to find a solution whose quality is not much worse than an optimal solution. To do so, we transform this problem setting into a weighted $xy$-monotone path problem, between $\sPoint$ and $\tPoint$, in the free-space diagram of $T_1$ and $T_2$. To measure the sum of the lengths of the subcurves of $T_1$ and $T_2$ directly in the free-space diagram, we stretch and compress the columns and rows of the diagram, such that their widths and heights are equal to the lengths of the corresponding segments. We call the resulting diagram the \emph{deformed free-space diagram} and  denote it by $F$. 
To solve the \MinEx\ (respectively, \MaxIn) problem,  we look for an $xy$-monotone path  $\pi_{\sPoint \tPoint} \subset F$
from $\sPoint$ to $\tPoint$, where our goal is to minimize (respectively, maximize) the length of $\pi_{\sPoint \tPoint}$, lying in the forbidden (respectively, free) space of $F$. Note that the length of the polygonal curves $T_1$ and $T_2$ corresponding to 
the parts of $\pi_{\sPoint \tPoint}$ lying in the forbidden or the free space equals the length of  $\pi_{\sPoint \tPoint}$ measured under $L_1$-metric in the corresponding space. We have the following observation. 

\begin{observation}\label{obs:transformation}
Let $T_1$ and $T_2$ be two arbitrary polygonal curves in $\mathbb{R}^2$ and let $F$ be the corresponding deformed free-space diagram for a leash length of $\varepsilon$. Let $\pi_{\sPoint \tPoint} \subset F$ be a path corresponding to a pair of parameterizations $\left( \alpha_1,\alpha_2 \right)$ of $T_1$ and $T_2$ w.r.t. $\varepsilon$. Then, the sum of the lengths of the forbidden (respectively, free) paths of $\pi_{\sPoint \tPoint}$, measured under $L_1$-metric, is equal to $\quality{\alpha}{B}$ (respectively, $\quality{\alpha}{W}$).
\end{observation}

Now it is easy to see that the \MinEx\ and \MaxIn\ problems are transformed to the following path problems.\\
\noindent {\bf Weighted shortest $xy$-monotone path (\ShortMP) problem}: Compute an $xy$-monotone weighted shortest path from 
 $\sPoint$ to $\tPoint$ in $F$, where the weight in the forbidden-space is one and the weight in the free-space is zero. The length of a path is defined as the sum of the lengths (measured in $L_1$-metric) of the part of the path lying in the forbidden space.\\
\noindent {\bf Weighted longest $xy$-monotone path (\LongMP) problem}: Compute an $xy$-monotone weighted longest path from 
 $\sPoint$ to $\tPoint$ in $F$, where the weight in the forbidden-space is zero and the weight in the free-space is one. The length of a path is defined as the sum of the lengths (measured in $L_1$-metric) of the part of the path lying in the free space.


\subsection{New results}\label{subsec:outline}
In Section \ref{sec:unsolv}, we establish that the \MinEx\ and \MaxIn\ problems are not solvable exactly by radicals over $\mathbb{Q}$.
This is proved using Observation \ref{obs:transformation} and showing that the \ShortMP\ problem
is unsolvable within  the \emph{Algebraic Computation Model over the Rational Numbers} (\acmq). In this model we can compute exactly any number that can be obtained from the rationals $\mathbb{Q}$ by applying a finite number of operations from $+,-,\times,\div,\sqrt[k]{}$, for any integer $k\geq 2$ \cite{Bajaj88, dec2012}.
The proof is based on Galois theory. 
Motivated by that, we turn our attention towards approximation algorithms for the \MinEx\ and the \MaxIn\ problems. 

In Section \ref{sec:firstAlgorithm}, we transform the \MinEx\ problem to the \ShortMP\ problem, which in turn is transformed to a shortest path problem in directed acyclic graphs (Lemma \ref{lem:appPathG}).
We propose an algorithm that approximates the weighted $xy$-monotone shortest path up to an additive error. This error is related to the lengths of the curves $T_1$ and $T_2$.
The running time of this algorithm is $\mathcal{O} \left( \frac{n^4}{\delta^2} \right)$, where $\delta$ is the approximation parameter (Theorem \ref{thm:nonoptimalResult}).
This algorithm also provides an approximate solution for the \MaxIn\ problem and with the same approximation quality (Corollary \ref{cor:nonoptimalResult2}).

In Section \ref{sec:improvement}, we improve this running time to $\mathcal{O} \left( \frac{n^3}{\delta}   \log \left( \frac{n}{\delta} \right) \right)$ (Theorem \ref{thm:finalResult}). To do so, we solve a subproblem related to forming a `small' graph over a convex set of points that preserves $L_1$-distances between certain pairs of points.  
In Section \ref{sec:FPTAS}, we discuss why FPTAS for the \MinEx\ and \MaxIn\ problems may not be feasible. 
However,  for the \MaxIn\ problem, we are able to design a $(1-\delta)$-approximation algorithm running in polynomial time, and its complexity depends upon the size of the input $n$, approximation factor $\delta$, and an additional parameter $\gamma$ defined as follows.
Consider the \MaxIn\ problem in the setting where distances are measured in the  $L_1$-metric, i.e., the distance
between a pair of points, one on trajectory $T_1$ and other on trajectory $T_2$, is measured using the $L_1$-metric. It turns out that the
free space within a cell for a given leash length is still convex, but its boundary is composed of straight line segments, instead of that of ellipses.  We define $\gamma$ to be the length of the optimal solution for \MaxIn\ problem in $L_1$-metric. Furthermore, Buchin et al. \cite{Yusu09} have shown that $\gamma$ can be computed in polynomial time.

\section{Unsolvability of \MinEx\ and \MaxIn\ Problems}\label{sec:unsolv}
Observation \ref{obs:transformation} implies that solving
the \MinEx\ (respectively, \MaxIn) problem is equivalent to finding a weighted shortest (respectively, longest) $xy$-monotone path, connecting  $\sPoint$ and $\tPoint$ in $F$, where the forbidden (respectively, free) space is weighted with $1$ and the rest with $0$, i.e., solving the \ShortMP\ (respectively, \LongMP) problem.  The difficulty is that these weighted path problems need to be solved through obstacles which have curved (elliptical) boundaries. In Theorem \ref{thm:unsolvable} we prove that the \ShortMP\ problem is
not solvable within the \acmq.
In the \acmq,
the usual arithmetic operations and the extraction of $k$-th roots
for any positive integer $k$ are available at unit cost.
In this model, each storage location is capable of holding any element of $\mathcal{C}$, where $\mathcal{C}$ is the following set.
For all $q,q'\in\mathbb{Q}$,
the following numbers are elements of $\mathcal{C}$ for any positive integer $k$: $q$, $q+q'$, $q-q'$, $q\times q'$, $q\div q'$ and $\sqrt[k]{q}$.
Notice that $\mathcal{C}$ contains complex numbers.

We can think of this model of computation from two different angles:
the algebraic point of view and the computer science point of view.
Let $p_d(x) = 0$ be a polynomial equation of degree $d$ with coefficients in $\mathcal{C}$.
From classical algebra,
we know that if $d\leq 4$,
then all the solutions to this equation are elements of $\mathcal{C}$
and can therefore be computed in $O(1)$ time within the \acmq.
In other words,
when $d\leq 4$,
there is a formula to solve $p_d(x) = 0$
that involves a finite number of arithmetic operations and $k$-th roots.
In this case,
we say that $p_d$ is solvable by radicals.
From Galois theory,
we know that for any $d \geq 5$,
there exist $p_d$'s for which no solution to $p_d(x) = 0$ belong to $\mathcal{C}$.
Therefore,
if $d \geq 5$,
such an equation cannot be solved in general within the \acmq.
In other words,
when $d \geq 5$,
there is no general formula to solve $p_d(x) = 0$
that involves a finite number of arithmetic operations and $k$-th roots.
In this case,
we say that $p_d$ is not solvable by radicals.
However,
there are some polynomial equations of degree $d \geq 5$ that can be solved by radicals.

The solvability of $p_d(x)=0$ by radicals
is determined by its \emph{Galois group} $\textrm{Gal}(p_d)$.
Indeed,
when $p_d$ is irreducible,
$p_d(x) = 0$ is solvable by radicals if and only if $\textrm{Gal}(p_d)$ is a \emph{solvable group}
(refer to~\cite{Dummit03} for an introduction to solvable groups).
Let us look at an example.
Consider $p_4(x) = x^4-2x^2+9$.
There is a general formula for solving quartic equations,
but we will solve it in a way that gives an intuition of what is a Galois group
and what is a solvable group.
We observe that $p_4(x) = \left(x^2\right)^2-2\left(x^2\right)^1+9$.
Therefore,
$x^2 = 1\pm 2i\sqrt{2}$,
from which $x = \pm\sqrt{1\pm 2i\sqrt{2}} = \pm i \pm\sqrt{2}$.
There are two steps in this solution.
At the beginning of the first step,
we have $p_4(x)$ that is a polynomial with coefficients in the field $\mathbb{Q}$.
At the end of the first step,
we have two quadratic polynomials: $x^2 - \left(1 \pm 2i\sqrt{2}\right)$.
These polynomials have coefficients in the field $\mathbb{Q}\left[i\sqrt{2}\right]$,
where $\mathbb{Q}\left[i\sqrt{2}\right]$ is the smallest field that contains $\mathbb{Q}$ and $i\sqrt{2}$.
At the end of the second step,
we have four linear polynomials:
$x - \left(\pm i \pm\sqrt{2}\right)$.
These polynomials have coefficients in the field $\mathbb{Q}\left[i,\sqrt{2}\right]$,
where $\mathbb{Q}\left[i,\sqrt{2}\right]$ is the smallest field that contains $\mathbb{Q}$ and $i$ and $\sqrt{2}$.
The field $\mathbb{Q}\left[i,\sqrt{2}\right]$ is also the smallest field that contains all the roots of $p_4$.
It is called the \emph{splitting field} of $p_4$.
Notice that $\mathbb{Q}$ is a subfield of $\mathbb{Q}\left[i\sqrt{2}\right]$,
which is a subfield of $\mathbb{Q}\left[i,\sqrt{2}\right]$.
Intuitively,
to solve a polynomial equation means finding such a chain of fields.

The Galois group $\textrm{Gal}(p_n)$ of an irreducible polynomial $p_n$ is the group of automorphisms of the splitting field $\mathbb{F}$ of $p_n$.
In our example,
$\mathbb{F} \cong \mathbb{Q}\left[i,\sqrt{2}\right]$
and $\textrm{Gal}(p_4) \cong \langle \sigma,\tau \mid \sigma^2 = \tau^2 = (\sigma\tau)^2 = 1 \rangle \cong V_4$,
where $V_4$ is the Klein group.
We can define $\sigma(i) = -i$, $\sigma(\sqrt{2}) = \sqrt{2}$, $\tau(i) = i$, $\tau(\sqrt{2}) = -\sqrt{2}$
and $\sigma(x) = \tau(x) = x$ for all $x\in\mathbb{Q}$.
There is a one-to-one correspondence between the lattice of subfields of $\mathbb{F}$
and the lattice of subgroups of $\textrm{Gal}(p_n)$.
Indeed,
for each subgroup $G$ of $\textrm{Gal}(p_n)$,
there is one subfield $\mathbb{K}$ of $\mathbb{F}$ that is \emph{fixed} by $G$
and vice-versa.
In our example,
$\mathbb{Q}[i]$ is fixed by $\langle 1,\tau \rangle$.
Indeed,
for any $x\in \mathbb{Q}[i]$, $\tau(x) = x$.
The field $\mathbb{Q}[\sqrt{2}]$ is fixed by $\langle 1,\sigma \rangle$,
$\mathbb{Q}[i\sqrt{2}]$ is fixed by $\langle 1,\sigma\tau \rangle$
and $\mathbb{Q}[i,\sqrt{2}]$ is fixed by $\langle 1 \rangle$.
The chain of fields that corresponds to the solution of a polynomial equation
can therefore be thought of as a chain of groups.

When we solve a polynomial equation by radicals,
we travel across a chain of fields that satisfies the following property.
If $\mathbb{K}_1$ and $\mathbb{K}_2$ are two consecutive fields in the chain,
then $\mathbb{K}_2 \cong \mathbb{K}_1[\sqrt[k]{\alpha}]$ for a positive integer $k$ and an $\alpha\in \mathbb{K}_1$.
We can prove
(refer to~\cite{Dummit03})
that in this case,
the corresponding chain of groups satisfies the following property.
If $G_1$ and $G_2$ are two consecutive groups in the chain,
then $G_1$ is a normal subgroup of $G_2$.
In this case,
we say that $\textrm{Gal}(p_n)$ is solvable.
In our example,
$\langle 1\rangle$ is a normal subgroup of $\langle 1,\sigma\tau\rangle$
which is a normal subgroup of $V_4$.

In the field of computer science,
the \acmq was first studied by Bajaj~\cite{Bajaj88}
and the name \emph{Algebraic Computation Model over the Rational Numbers} (\acmq)
was introduced by De Carufel et al.~\cite{dec2012}.
Bajaj proved that the Fermat-Weber problem cannot be solved within the \acmq.
He established a criteria
(refer to Lemma~\ref{lemma bajaj})
that helps deciding whether $\textrm{Gal}(p_d)$ is solvable.
Let $S_d$ be the symmetric group over $d$ elements.
Bajaj's criteria concludes that $\textrm{Gal}(p_d) \cong S_d$
if it is the case
or does not conclude otherwise.
The key observation is that $S_d$ is solvable if and only if $d \leq 4$.
Hence,
to prove that a problem cannot be solved within the \acmq,
it suffices to find an instance that leads to a polynomial equation $p_d$
such that $\textrm{Gal}(p_d) \cong S_d$
(with $d\geq 5$).
The following simplified version of Bajaj's lemma appeared in~\cite{dec2012},
\begin{lemma}[Bajaj]
\label{lemma bajaj}
Let $p_d$ be a polynomial of even degree $d\geq 6$.
Suppose that there are three prime numbers $q_1$, $q_2$ and $q_3$ that do
not divide the discriminant $\Delta(p_d)$ of $p_d$ and such that
\begin{eqnarray}
\label{thm bajaj cond 1}
p_d(x) &\equiv& \overline{p}_{d}(x) \pmod{q_1} \enspace,\\
\label{thm bajaj cond 2}
p_d(x) &\equiv& \overline{p}_1(x)\overline{p}_{d-1}(x) \pmod{q_2} \enspace,\\
\label{thm bajaj cond 3}
p_d(x) &\equiv& \overline{p}_1'(x)\overline{p}_2(x)\overline{p}_{d-3}(x) \pmod{q_3} \enspace,
\end{eqnarray}
where $\overline{p}_d(x)$ is an irreducible polynomial of degree $d$ modulo $q_1$;
$\overline{p}_{d-1}(x)$ (respectively $\overline{p}_1(x)$) is an irreducible polynomial of degree
$d-1$ (respectively of degree $1$) modulo $q_2$;
$\overline{p}_{d-3}(x)$ (respectively $\overline{p}_1'(x)$ and $\overline{p}_2(x)$) is an irreducible polynomial of degree
$d-3$ (respectively of degree $1$ and of degree $2$) modulo $q_3$. Then
$\textrm{Gal}(p_d) \cong S_d$.

If $d\geq 5$ is odd,
the same result holds if we replace (\ref{thm bajaj cond 3}) by
\begin{eqnarray}
\label{thm bajaj cond 4}
p_d(x) &\equiv& \overline{p}_2(x)\overline{p}_{d-2}(x) \pmod{q_4} \enspace,
\end{eqnarray}
where $q_4$ is a prime number such that $q_4\not|\Delta(p_d)$
and $\overline{p}_{d-2}(x)$ (respectively $\overline{p}_2(x)$) is an irreducible polynomial of degree
$d-2$ (respectively of degree $2$) modulo $q_4$.
\end{lemma}
Observe that 
\eqref{thm bajaj cond 1} implies that $p_d(x)$ is irreducible,
which implies that $\textrm{Gal}(p_d)$ is a \emph{transitive} group.
(\ref{thm bajaj cond 2}) and (\ref{thm bajaj cond 3})
guarantee the existence of a $(d-1)$-cycle
and of an element with cycle decomposition $(2,d-3)$ in $\textrm{Gal}(p_d)$.
These two elements,
together with the transitivity of $\textrm{Gal}(p_d)$,
imply that $\textrm{Gal}(p_d)\cong S_d$.

One of the usual models of computation for computational geometers is the real-RAM model
(refer to Preparata and Shamos for instance).
This model is more general than the \acmq.
Indeed,
it enables to manipulate any real number
and provides transcendental functions (at unit cost),
such as: trigonometric and logarithmic functions.
However,
a significant amount of classical problems in computational geometry can be solved within the \acmq
since their solution involves only the arithmetic operations and the square root.

The strategy to prove Theorem~\ref{thm:unsolvable} is as follows. 
We provide an example of two trajectories for which
both the length and the coordinates of the bending points of any weighted shortest $xy$-monotone path cannot be computed by radicals. We reduce such a computation to the solution of a polynomial equation of degree $8$. We show that the Galois group of this polynomial is isomorphic to $S_8$.
\begin{theorem}\label{thm:unsolvable}
Let $T_1$ and $T_2$ be two polygonal curves. Denote by $\sPoint$ (respectively, by $\tPoint$) the bottom left corner (respectively, the upper right corner) of their deformed free-space diagram F. Let $\pi_{\sPoint \tPoint}$ be any weighted shortest $xy$-monotone path (with respect to $L_1$-metric) from $\sPoint$ to $\tPoint$. In this setting, the \ShortMP\ problem is unsolvable within the \acmq, i.e.,  in general, both the length and the coordinates of the bending points of $\pi_{\sPoint \tPoint}$ cannot be computed by radicals.
\end{theorem}
\begin{proof}
Take $T_1=abc$, $T_2=de$ and $\varepsilon = 1$, where $a=(0,0)$, $b=(1,0)$, $c=\left(-1,-\frac{31}{240}\right)$, $d=\left(-\frac{1}{2},\frac{3}{4}\right)$, $e=\left(\frac{5}{2},\frac{13}{8}\right)$ (see Figure~\ref{figure example}a). Hence, in the deformed free-space diagram, $\sPoint=(0,0)$ and $\tPoint=\left(\frac{721}{240},\frac{25}{8}\right)$ (refer to Figure~\ref{figure example}b).
\begin{figure}[ht]
	\begin{center}
		\begin{tabular}{ccc}
			\includegraphics[height=3cm]{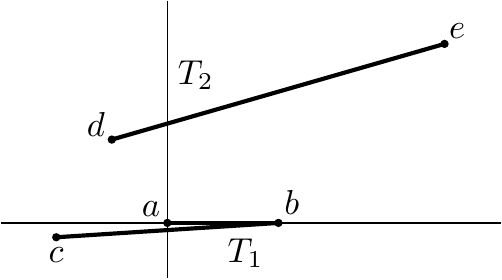}&&\includegraphics[height=5cm]{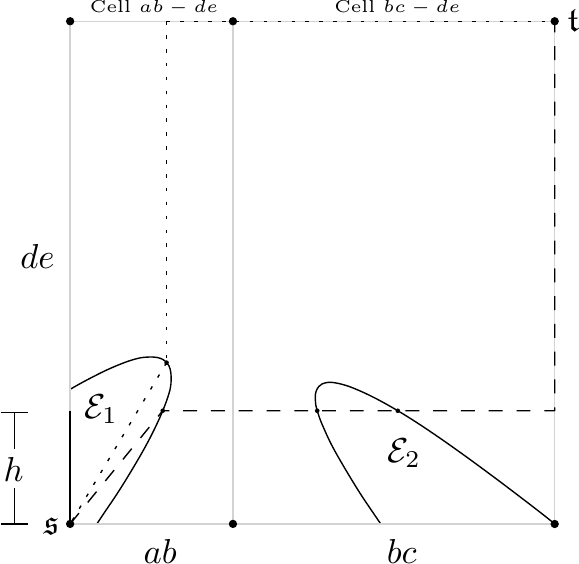}\\
			{(a)}&&{(b)}\\
		\end{tabular}
	\end{center}
	\caption{a) Two polygonal curves, $T_1$ and $T_2$. b) The deformed free-space diagram $F$ for $T_1$ and $T_2$. The dashed line is a weighted shortest $xy$-monotone path from $\sPoint$ to $\tPoint$. The dotted line is a weighted shortest $xy$-monotone path that crosses $\mathcal{E}_1$ but not $\mathcal{E}_2$.}
	\label{figure example}
\end{figure}
The parametric equations of  $ab$, $bc$ and $de$ are respectively
\begin{eqnarray*}
ab:a+\frac{u}{|ab|}(b-a) &=& u(1,0) \qquad (0 \leq u \leq 1),\\
bc:b+\frac{u}{|bc|}(c-b) &=& (1,0)+\frac{u}{481/240}\left(-2,-\frac{31}{240}\right) \qquad \left(0 \leq u \leq \frac{481}{240}\right),\\
de:d+\frac{u}{|de|}(e-d) &=& \left(-\frac{1}{2},\frac{3}{4}\right)+\frac{u}{25/8}\left(3,\frac{7}{8}\right) \qquad \left(0 \leq u \leq \frac{25}{8}\right).
\end{eqnarray*}

In cell $ab$-$de$,
we have the ellipse
$\mathcal{E}_1 : \left(x - \frac{24}{25}\,y + \frac{1}{2}\right)^2 + \left(\frac{7}{25}\,y + \frac{3}{4}\right)^2 = 1,$
where $0 \leq x \leq \frac{1}{2}$ and $0 \leq y \leq \frac{25}{28}$.
In cell $bc$-$de$,
we have the ellipse
$\mathcal{E}_2:\left(\frac{480}{481}\,(x-1) + \frac{24}{25}\,y - \frac{3}{2}\right)^2 
+ \left(\frac{31}{481}\,(x-1) + \frac{7}{25}\,y + \frac{3}{4}\right)^2 = 1,$
where $\frac{8599}{5232} \leq x \leq \frac{721}{240}$ and $0 \leq y \leq \frac{3725}{5232}$.
Since $\sPoint\in\mathcal{E}_1$,
then either (1) $\pi_{\sPoint \tPoint}$ crosses $\mathcal{E}_1$ but not $\mathcal{E}_2$ or (2)
$\pi_{\sPoint \tPoint}$ crosses both $\mathcal{E}_1$ and $\mathcal{E}_2$.\\
\noindent 1.  If $\pi_{\sPoint \tPoint}$ crosses $\mathcal{E}_1$ but not $\mathcal{E}_2$,
we find the following optimal path
by elementary calculus.
Travel in $\mathcal{E}_1$ from $\sPoint$ to
$\left(\frac{1}{14}\left(35\sqrt{2}-43\right),\frac{5}{28}\left(14\sqrt{2} - 15\right)\right)$
and then to $\tPoint$
(outside of $\mathcal{E}_1$ and outside of $\mathcal{E}_2$).
The length of this path is $\frac{1}{240}\left(2851 - 1200\sqrt{2}\right) \approx 4.80810.$\\
2. If $\pi_{\sPoint \tPoint}$ crosses both $\mathcal{E}_1$ and $\mathcal{E}_2$,
then $\pi_{\sPoint \tPoint}$ must exit $\mathcal{E}_1$ at the same height it enters $\mathcal{E}_2$
since $\mathcal{E}_1$ is inclined towards $\mathcal{E}_2$.
Let $h$ be this height
and let $h'$ be the height at which $\pi_{\sPoint \tPoint}$ exits $\mathcal{E}_2$.
By elementary calculus,
we find that for an $xy$-monotone path to be shortest,
we need to have $h = h'$.

The length of such a shortest path can be expressed in the following way:
\begin{eqnarray*}
\frac{1591}{240} - \frac{49}{25}\,h - \frac{\sqrt{4375 - 4200 h - 784 h^2}}{100} - 
 \frac{\sqrt{165296875 - 212680800 h - 27373824 h^2}}{12025} \enspace,
\end{eqnarray*}
for an $h$ to be determined.
If we look for the values of $h$
for which the derivative of this last expression is $0$,
we find that $h$ must be a solution of
$$
\begin{array}{l}
p(h) = 585090042379589947534557557525634765625 - 3039825965000401080955586792871093750000 \, h \\ 
+ 5307213095548843266935155031210937500000 \, h^2 - 2973595218630130711131340711267500000000 \, h^3 \\ 
- 649444075888789852190828979088700000000 \, h^4 + 562445109533777824218782819614464000000 \, h^5 \\ 
+ 193996238215889538903991144689745920000 \, h^6 + 21705929355568145355212682312548352000 \, h^7 \\ 
+ 826789346560923302640987287586865152 \, h^8 = 0.
\end{array}
$$
By numerical methods,
we find $h \approx 0.50696$
and the length of the corresponding path is approximately $4.59277$,
so it is a global weighted shortest $xy$-monotone path.

The discriminant of $p$ is $\Delta(p) = 2^{226}\cdot 3^{28}\cdot 5^{115}\cdot 7^{28}\cdot 13^{36}\cdot 23^2\cdot 29^2\cdot 31^6\cdot 37^{36}\cdot 47^3\cdot 53^3\cdot 109^{24}\cdot 151^2\cdot 281\cdot 443^3\cdot 467^3\cdot 2909\cdot 3313\cdot 18959^2\cdot 1120001\cdot 33513959\cdot 89206609^2\cdot 261977539^2\cdot 28587810523 \cdot 306854901568937582895921655033\cdot 5739056544236116407796954338317251465127$.
We have
\begin{eqnarray*}
p(h) &\equiv& 54h^8+51h^7+7h^6+78h^5+50h^4+95h^3+84h^2+47h+59 \pmod{101} \enspace,\\
p(h) &\equiv& 13(h+11)(h^7+14h^6+16h^5+16h^3+9h^2+11h+9) \pmod{17} \enspace,\\
p(h) &\equiv& 64(h+52)(h^2+9h+42)(h^5+44h^4+7h^3+21h^2+31h+16) \pmod{71} \enspace.
\end{eqnarray*}
Hence,
by Lemma~\ref{lemma bajaj},
$\textrm{Gal}(p) \cong S_8$
which is not solvable.
Hence,
$p(h)=0$ is not solvable by radicals.
Consequently,
both the length of $\pi_{\sPoint \tPoint}$
and the coordinates of its bending points
cannot be computed by radicals.
\end{proof}
Combining this with Theorem \ref{thm:unsolvable} and Observation \ref{obs:transformation} we have the following result.
\begin{cor}
It is not possible to design an algorithm that can exactly solve the \MinEx\ (or \MaxIn)  problem within
the \acmq.
\end{cor}
In the  proof of Theorem \ref{thm:unsolvable}, we show that for $n=2$, we can construct examples where we  have to solve a polynomial equation of degree $8$.
In general, we can construct examples for which  the degree of the polynomial equations involved is $\Omega(n)$.
Therefore, we cannot suppose that we are in a model of computation where polynomial equations of bounded degree can be solved in constant time. 

\section{An Approximation Algorithm 
}\label{sec:firstAlgorithm}
In this section, we present an approximation algorithm with an additive error for the  \MinEx\ problem and we will show that the computed approximate solution is an approximate solution for the \MaxIn\ problem as well.
The input to the problem consists of two polygonal curves $T_1$ and $T_2$, an arbitrary fixed leash 
length $\varepsilon \geq 0$ and an approximation parameter $\delta >0$. We want to compute a pair of parameterizations $\left( \tilde\alpha_1 , \tilde\alpha_2 \right)$ and its quality $\quality{\tilde{\alpha}}{B}$ for $T_1$ and $T_2$, such that $\quality{\tilde{\alpha}}{B}$ is a good approximation of $\Tquality{T}{B}$. We also want to construct two polygonal curves, $T_1'$ and $T_2'$ such that $\delta_{F} \left( T_1',T_2' \right) \leq \varepsilon$, that correspond to a solution for the \MinEx\  problem.
We abbreviate the deformed free-space diagram by $F$, its free-space by $\textit{W}$ and the forbidden-space by $\textit{B}$.
Recall that $F$ consists of $\mathcal{O}(n^2)$ cells and each cell is a rectangle, whose free space is (a portion of) an ellipse.
Our approach is as follows.

We have seen in Section \ref{subsec:ProblemDefinition} how to transform the \MinEx\ problem into the \ShortMP\ problem in the deformed free space diagram $F$. To design an approximation algorithm for the \ShortMP\ problem, we will define a graph $G$ over $F$, and show that for each path $\pi_{\sPoint \tPoint}$ in $F$, there exists a path $\widetilde{\pi}_{\sPoint \tPoint}$ in $G$, which stays  close to $\pi_{\sPoint \tPoint}$. Thus, paths in $G$  approximates $xy$-monotone weighted shortest paths in $F$.  Once we have a shortest path in $G$, it will be fairly straightforward to embed this path in $F$, and deduce an approximate  solution for the \MinEx\ problem.  Though we are casting a geometric problem into a combinatorial setting, to simplify our notation, we may
refer to a point $p \in F$ as also a vertex $p$ in $G$ -  the meaning will be clear from the context.  We construct $G$ as follows.\\
{\bf Step 1: Construct $F$.} First compute the free-space diagram using the algorithm of Alt and Godau \cite{Alt95}, then stretch the columns/rows of the diagram, such that widths/heights are equal to the length of the corresponding segments to obtain the deformed free-space diagram $F$. \\ 
{\bf Step 2: Construct Grid.} Add $\frac{n}{\delta}$ additional equidistant vertical and horizontal \emph{grid lines} to $F$ (Figure \ref{fig:arrangement}). 
Find the intersection of every vertical grid line $\ell_{v_i}$ with the boundary of each ellipse. For each of these intersection points, add a new horizontal \emph{intersection line}, passing through that point.
Perform analogous steps for every horizontal grid line $\ell_{h_i}$. For each intersection point add a vertical intersection line, passing through that point.\\ 
{\bf Step 3: Construct $G$.} Compute the arrangement $A$ induced by all of the grid lines, the intersection lines and the boundary of ellipses. The vertices of $G$ are the vertices in $A$. For each edge $(p,q)$ in $A$, if either $\overrightarrow{pq}$ or $\overrightarrow{qp}$ is $xy$-monotone then add the corresponding directed edge into $G$. The weight of this edge is equal to its length in $L_1$-metric if it is lying in the forbidden space, otherwise it is zero.
\begin{figure}[ht] 
		\centering
		\includegraphics[scale=.9]{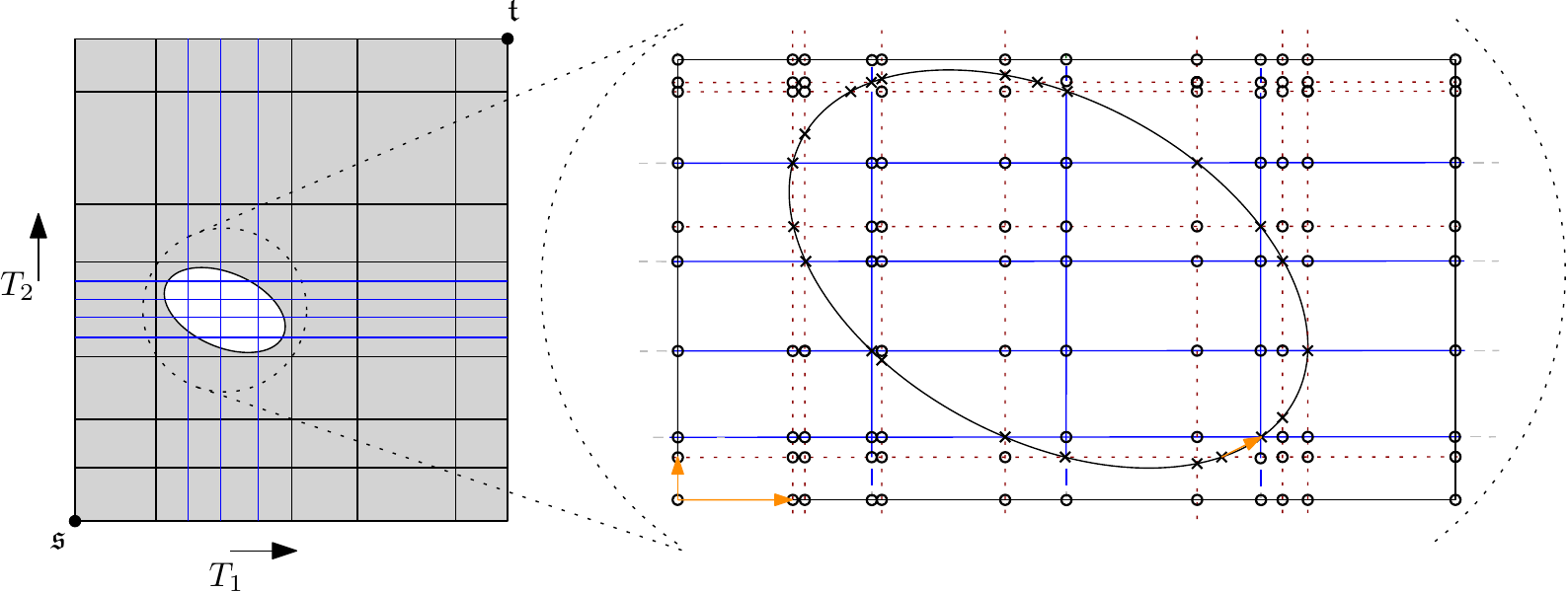} 
    		\caption{An arrangement of lines and the boundary of an ellipse. Blue solid lines are grid lines and dotted dark red lines are intersection lines. Crosses represent vertices on the boundary of the ellipse resulting from intersections of the ellipse with the lines. Circles show some of the vertices on the intersection of (grid and intersection) lines. All vertices on a line are connected by directed edges (colored orange arrows) preserving $xy$-monotone ordering. Also, two vertices on the boundary of the ellipse are joined by an edge (also colored orange), if they are consecutive and the edge is $xy$-monotone.}
   		\label{fig:arrangement}
\end{figure}

Observe that $G$ is acyclic as all of its edges are directed and are $xy$-monotone.
Compute a weighted shortest path from $\sPoint$ to $\tPoint$ in $G$, and output its corresponding geometric embedding as the desired approximate solution. In the next three lemmas,  we establish that $G$ can be used to 
provide an approximation algorithm for the  \ShortMP\ problem. First we need the following definition. 
A vertex $s'$ is \emph{directly dominated} by the point $s \in F$, if and only if, $s'$ is dominated by $s$ and there exists no vertex of $G$ in the interior of  $ss'$.

%

%
%
\newcounter{pathinG}
\setcounter{pathinG}{\value{lemcounter}}
\begin{lemma}\label{lem:subPath}
	Let $\pi_{st} \subset \textit{W}$ be an $xy$-monotone shortest path from $s$ to $t$, connecting $s$ on a grid line $\ell_s$ with $t$ on another grid line $\ell_t$ inside a parameter cell. Furthermore, let $s' \in \ell_s$ (respectively, $t' \in \ell_t$) be a vertex from $G$ directly dominated by $s$ (respectively, $t$). Then, there exists a path $\widetilde{\pi}_{s't'}$ in $G$ from $s'$ to $t'$, such that $\widetilde{\pi}_{s't'} \subset \textit{W}$.
\end{lemma}
\begin{proof}	Since $s \in \textit{W}$ and $s'$ is directly dominated by $s$, we know $s' \in \textit{W}$ ($s$ could be equal to $s'$). We discuss the following four cases:
	\begin{enumerate}
\item[a)] $s$ and $t$ lie on two vertical lines $\ell^v_{i}$ and $\ell^v_{i+1}$ (Figure \ref{fig:subPath}a)).\\ Let $a$ be the intersection point of $\ell_{s'}$ with $\ell^v_{i+1}$. We walk on the horizontal line $\ell_{s'} \ni s'$, until we leave $\mathcal{E}$ at a point $x\in \partial\mathcal{E}$ or we reach $a$. If we encounter $a$ first, we are done. Assume  we first reach  $x$. This implies that the gradient of $\partial \mathcal{E}$ in $x$ is positive and so its tangent line $T_x$ has to be $xy$-monotone. Since $\pi_{st} \subset \textit{W}$, it follows that $\partial \mathcal{E}$ is squeezed between $T_x$, $\pi_{st}$, $\ell^v_i$ and $\ell^v_{i+1}$. We denote this region by $R$. The boundary $\partial \mathcal{E}$ cannot turn back, as this would contradict  the convexity of $\mathcal{E}$. This implies that we can start at  $x$, follow $\partial \mathcal{E}$ using only $xy$-monotone edges, stay in $R$ and reach $\ell_{i+1}^v$ at a point $y\in R$. Since all the utilized edges lie in $\textit{W}$, the final path is in $\textit{W}$.
\item[b)] $s$ and $t$ lie on two horizontal lines $\ell^h_{i}$ and $\ell^h_{i+1}$ (Figure \ref{fig:subPath}b)).\\ 
Use the same argument as in a), but first walk upwards.
\item[c)] $s \in \ell^v_i$ and $t \in \ell^h_j$ (Figure \ref{fig:subPath}c)).\\ Use the same argument as in a), but first move to $s$ and then walk upwards.
\item[d)] $s \in \ell^h_i$ and $t \in \ell^v_j$ (Figure \ref{fig:subPath}d)).\\ Use the same argument as in a), but first move to $s$ and then walk to the right.
	\end{enumerate} 
	After reaching the point $y \in \partial \mathcal{E} \cap \ell_t$, walk on $\ell_t$ from $y$ towards $t$, until the closest vertex $t'$ lying below and to the left of $t$ is reached. Since $t,y \in \textit{W}$, that guarantees again $t' \in \textit{W}$.
\end{proof}

\begin{figure}[ht]
    		\begin{center}
     			 \begin{tabular}{cccc}
       				\includegraphics[height=3cm]{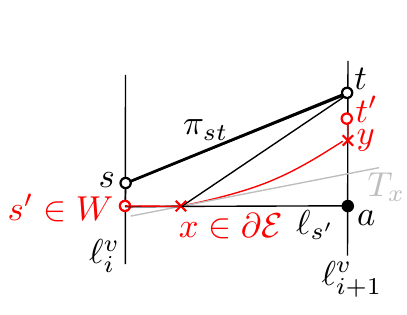} & \includegraphics[height=3cm]{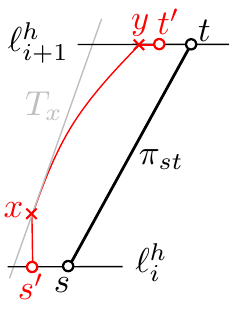} & \includegraphics[height=3cm]{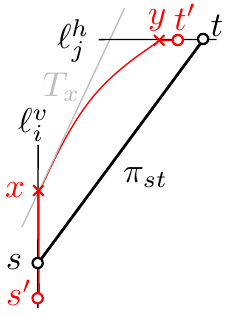} & \includegraphics[height=3cm]{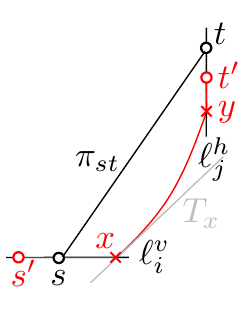}\\
       				{\small a) $s \in \ell^v_{i}$ and $t \in \ell^v_{i+1}$.}  & {\small b) $s \in \ell^h_{i}$ and $t \in \ell^h_{i+1}$.} & {\small c) $s \in \ell^v_i$ and $t \in \ell^h_j$.} & {\small d) $s \in \ell^h_i$ and $t \in \ell^v_j$.}\\
      			\end{tabular}
    		\end{center}
    		\vspace*{-12pt}
    		\caption{Four cases of configurations for $s$ and $t$, and their corresponding grid lines that are used in the proof  of Lemma \ref{lem:subPath}.}
   		\label{fig:subPath}
\end{figure}

For the next lemma we need the following notation. Each parameter cell $C^{i,j}$ of $F$ defines a rectangle, having the left side $C_L^{i,j}$, the right side $C_R^{i,j}$, the bottom side $C_B^{i,j}$ and the top side $C_T^{i,j}$. Furthermore, let $||\pi||$ be the weighted length of a path $\pi$ in $F$, measured using the $L_1$-metric. Let $T$ be a polygonal curve composed of $k$ segment, $T=(t_0 t_1 , t_1 t_2 ,\mathellipsis, t_{k-1} t_k)$, then the length of $T$, $|T|=\sum\limits_{i=1}^k |t_{i-1} t_{i}|$.

\begin{lemma}\label{lem:cellSubPath}
	Let $C^{i,j}$ be an arbitrary parameter cell of  $F$. Let $\pi_{st}$ be a shortest $xy$-monotone path in $C^{i,j}$, connecting a point $s \in C_{L}^{i,j} \cup C_B^{i,j}$ with an arbitrary point $t \in C_{T}^{i,j} \cup C_R^{i,j}$. Let $\ell_s$ (respectively, $\ell_t$) be the side of $C^{i,j}$ on which $s$ (respectively, $t$) lies (see Figure \ref{fig:cellSubPath}). Let $s' \in \ell_s$ (respectively, $t' \in \ell_t$) be a vertex directly dominated by $s$ (respectively, $t$). Then, there exists an $xy$-monotone path $\widetilde{\pi}_{s't'}$, such that $||\widetilde{\pi}_{s't'}|| \leq ||\pi_{st}|| + 8   \frac{\delta}{n} \cdot\max \left\lbrace |T_1|, |T_2| \right\rbrace$.

\end{lemma}
\begin{proof}
	We will construct a path $\widetilde{\pi}_{s't'} \subset G$ such that it connects $s'$ and $t'$,  and it passes through $\textit{W}$ whenever $\pi_{st}$ passes through $\textit{W}$. This guarantees that the weighted length of $\widetilde{\pi}_{s't'}$, in $L_1$-metric, is not ``much'' bigger than that the weighted length of $\pi_{st}$. Let $\mathcal{E}$ be the ellipse describing $\textit{W}$ in $C^{i,j}$. Since $\mathcal{E}$ is convex and $\pi_{st}$ is a shortest path, it implies that there exists at most one entrance point $a$ into $\mathcal{E}$ and at most one exit point $b$ from $\mathcal{E}$.
	
	\begin{figure}[!btp]
		\centering
		\includegraphics[scale=1]{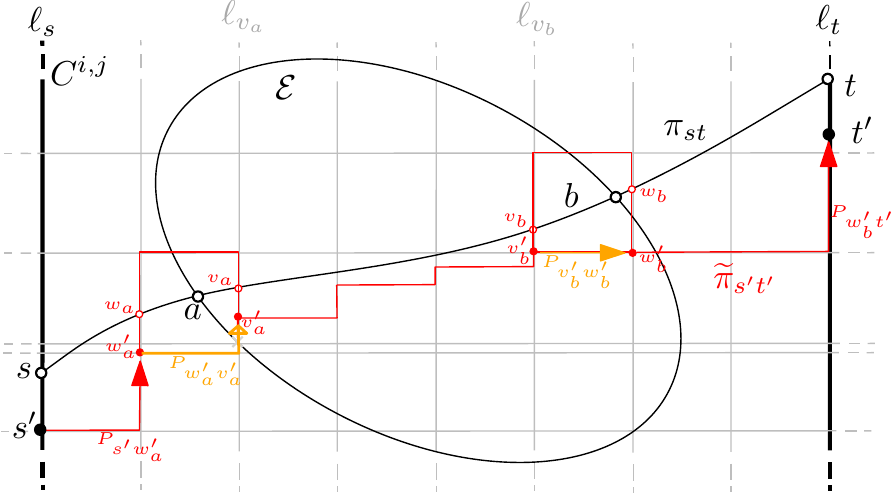}%
		\caption{Illustration of the proof of Lemma \ref{lem:cellSubPath}.}%
		\label{fig:cellSubPath}
	\end{figure}	
Let $v_a$ be the first point on $\pi_{st}$ after $a$ (w.r.t. $\pi_{st}$) that is on a horizontal or vertical grid line $\ell_{a}$ (see Figure \ref{fig:cellSubPath}).
Let $v_a'$ be the vertex directly dominated by $v_a$. Analogously, we define $v_b$ as the last intersecting point on $\pi_{st}$ with a horizontal or vertical grid line $\ell_b$, lying before $b$. We denote by  $v_b' \in \ell_b$ the vertex that is directly dominated by $v_b$. Let $\pi_1,...,\pi_k$ be the sub curves between $v_a$ and $v_b$ on $\pi_{st}$, separated by grid lines.
Their concatenation is the  curve $\pi_{v_av_b}$ connecting $v_a$ and $v_b$ on $\pi_{st}$. We know that  $\pi_{ab} \subset \textit{W}$, so $\pi_i \subset \textit{W}$ for all $i = 1,...,k$.  
Let $\widetilde{\pi}_{v_a'v_b'} \subset \textit{W} \cap G$ be the concatenation of the paths obtained by
applying Lemma \ref{lem:subPath} on each $\pi_{i}$.
The two remaining sub-curves of $\pi_{st}$ (one before and one after $\pi_{v_av_b}$) are passing through obstacles. The parts in the cells enclosing $a$ and $b$ are exceptions. Let $w_{a}$ be the entry point of $\pi_{st}$ to the grid subcell  that contains $a$. 
Let $w_{a}'$ be its left bottom corner vertex. Since $s$ and $s'$ lie in the same grid cell, it follows:
\begin{equation}
\label{eq:ssprimelenght}
|ss'|_x \leq \frac{\delta}{n} \cdot \max \left\lbrace |T_1|, |T_2| \right\rbrace,
\end{equation}
where $|ss'|_x$ denotes the length of the projection of $ss'$ on the $x$-axis.
Furthermore, $w_a'$ is dominated by $w_a$ and the weighted shortest connecting $xy$-monotone path $\pi_{s w_a}$ between $s$ and $w_a$ lies completely inside $\textit{B}$. Thus, it follows that $|\widetilde{\pi}_{sw_a'}| \leq |\pi_{sw_a}|$. Combining this with (\ref{eq:ssprimelenght}), we obtain an $xy$-monotone weighted shortest connecting path $\widetilde{\pi}_{s'w_a'}$ from $s'$ to $w_a'$, such that $||\widetilde{\pi}_{s'w_a'}|| \leq ||\pi_{sw_a}|| +  2   \frac{\delta}{n} \cdot \max \left\lbrace |T_1|, |T_2| \right\rbrace$. We do not know how large the part of $\pi_{w_av_a}$  that passes through $\textit{B}$ or $\textit{W}$ is. Hence, we assume that $\pi_{w_av_a} \subset \textit{W}$. However $\pi_{w_av_a}$ is enclosed only by one grid cell (the one that contains $w_a'$ and $v_a'$). Thus, it follows that there exists an $xy$-monotone path $\widetilde{\pi}_{w_a' v_a'}$ (corresponding to $\pi_{w_av_a}$ and connecting $w_a'$ with $v_a'$), for which $||\widetilde{\pi}_{w_a'v_a'}|| \leq 2   \frac{\delta}{n} \cdot \max \left\lbrace |T_1|, |T_2| \right\rbrace$. In the same way, we construct the paths $\widetilde{\pi}_{v_b'w_b'}$ and $\widetilde{\pi}_{w_b't'}$, such that the concatenation $\widetilde{\pi}_{s't'}$ of all these four paths is not much longer compared to $||\pi_{st}||$, i.e., $||\widetilde{\pi}_{s't'}|| \leq ||\pi_{st}|| + 8   \frac{\delta}{n} \cdot \max \left\lbrace |T_1|, |T_2| \right\rbrace $.
\end{proof}

\begin{lemma}
\label{lem:appPathG}
	The graph $G$ has a complexity of $\mathcal{O} \left( \frac{n^4}{\delta^2} \right)$. If $\pi_{\sPoint \tPoint} \subset F$ (respectively, $\widetilde{\pi}_{\sPoint \tPoint} \subset G$) is a weighted shortest $xy$-monotone path in $F$ (respectively, $G$) then
			$||\pi_{\sPoint \tPoint}|| \leq ||\widetilde{\pi}_{\sPoint \tPoint}|| \leq ||\pi_{\sPoint \tPoint}|| + \delta \cdot \max \left\lbrace |T_1|, |T_2| \right\rbrace$.
\end{lemma}
\begin{proof}
        Denote the sequence of grid cells that $\pi_{\sPoint \tPoint}$ intersects from $\sPoint$ to $\tPoint$ by 
	 $c_1,...,c_k$.  Let  $\widetilde{\pi}_{\sPoint \tPoint}$ be the concatenation of the paths obtained by applying  Lemma \ref{lem:cellSubPath} with $\delta := \frac{\delta}{16}$ on each $c_i$. The error made in each cell is therefore upper bounded by 
      $\frac{\delta}{2n} \cdot \max \left\lbrace |T_1|, |T_2| \right\rbrace$. Since $\pi_{\sPoint \tPoint}$ passes through at most 
       $2 n$ cells, it follows that $||\pi_{\sPoint \tPoint}|| \leq ||\widetilde{\pi}_{\sPoint \tPoint}|| \leq ||\pi_{\sPoint \tPoint}|| + \delta \cdot \max \left\lbrace |T_1|, |T_2| \right\rbrace$.
      Each of the $\frac{16   n}{\delta}$ grid lines can intersect at most $2 n$ ellipses, thus requiring the addition of at most $\frac{32   n^2}{\delta}$ additional intersection lines. So, the arrangement of all these lines has a complexity of $\mathcal{O} \left( \frac{n^4}{\delta^2} \right)$.
\end{proof}

The above lemma shows that for the \ShortMP\ problem, 
$\widetilde{\pi}_{\sPoint \tPoint} \subset G$ approximates $\pi_{\sPoint \tPoint} \subset F$. 
Next, we show how to derive an approximate solution for the \MinEx\ problem, given $\widetilde{\pi}_{\sPoint \tPoint}$.
The path $\widetilde{\pi}_{\sPoint \tPoint}$ passes through a sequence of parameter cells. For each edge in $\widetilde{\pi}_{\sPoint \tPoint}$, find its embedding in $F$. Since each of these embedded edges are $xy$-monotone and the end vertex of one edge is the start vertex of the next one, this results in an $xy$-monotone path 
from $\sPoint$ to $\tPoint$ in $F$. 
Let $\pi_{\sPoint \tPoint}$ be the weighted shortest $xy$-monotone path connecting $\sPoint$ to $\tPoint$. From Lemma \ref{lem:appPathG} it follows that $||\pi_{\sPoint \tPoint}|| \leq ||\widetilde{\pi}_{\sPoint \tPoint}|| \leq ||\pi_{\sPoint \tPoint}|| + \delta \cdot \max \left\lbrace |T_1|, |T_2| \right\rbrace$. 
Since $\widetilde{\pi}_{\sPoint \tPoint}$ is a concatenation of segments, it can be directly transformed into two corresponding parameterizations $\tilde\alpha_1$ and $\tilde\alpha_2$. Since $||\pi_{\sPoint \tPoint}|| \leq ||\widetilde{\pi}_{\sPoint \tPoint}|| + \delta \cdot \max \left\lbrace |T_1|, |T_2| \right\rbrace$, the approximation quality of the solution $\left( \tilde\alpha_1, \tilde\alpha_2 \right)$ follows from Observation \ref{obs:transformation}. Let $\widetilde{\pi}_{ab}$ be a maximal subpath of $\widetilde{\pi}_{\sPoint \tPoint}$, passing through the forbidden-space and connecting the points $a$ and $b$, which lie both on the boundary of the free-space. These points correspond to two matchings $\left(a_1, a_2 \right)$ and $\left(b_1, b_2 \right)$ of two points lying on $T_1$ and $T_2$. Since both lie on the boundary of the free-space, it follows that $|a_1 a_2| = \varepsilon$ and $|b_1 b_2| = \varepsilon$. This implies that the Fr\'{e}chet distance between ${a_1 b_1}$  and ${a_2 b_2}$ is not greater than $\varepsilon$. We exchange the curve between $a_1$ and $b_1$ on $T_1$ by the segment ${a_1b_1}$ and we proceed analogously for the curve between $a_2$ and $b_2$ on $T_2$ using the segment ${a_2b_2}$. With these substitutions for all maximal subcurves of $\widetilde{\pi}_{\sPoint \tPoint}$, passing through the forbidden-space, we obtain two new polygonal curves $T_1'$ and $T_2'$, with $\delta_{F} \left( T_1',T_2' \right) \leq \varepsilon$. Observation \ref{obs:transformation} implies, that the sum of the lengths of the removed subcurves on $T_1$ and $T_2$ is equal to $\quality{\tilde{\alpha}}{B}$. So, the sum of the lengths of the substituted curves does not exceed $\Tquality{T}{B} + \delta \cdot \max \left\lbrace |T_1|, |T_2| \right\rbrace$. The running time follows directly from Lemma \ref{lem:appPathG}, by running the linear time algorithm for finding a shortest path in a directed acyclic graph.
We summarize our result for the \MinEx\ problem in the following theorem.

\begin{theorem}\label{thm:nonoptimalResult}
	Given two polygonal curves $T_1$ and $T_2$ in the plane, an arbitrary fixed $\varepsilon \ge 0$ and an approximation parameter $\delta >0$, we can compute in $\mathcal{O}( \frac{n^4}{\delta^2})$ time a pair of parameterizations $\left( \tilde\alpha_1 , \tilde\alpha_2 \right)$ and its quality $\quality{\tilde{\alpha}}{B}$ for $T_1$ and $T_2$, such that
			$\Tquality{T}{B} \leq \quality{\tilde{\alpha}}{B} \leq \Tquality{T}{B} + \delta \cdot \max \left\lbrace |T_1|, |T_2| \right\rbrace$.
Furthermore, we can construct two polygonal curves, $T_1'$ and $T_2'$, realizing $\quality{\tilde{\alpha}}{B}$, such that $\delta_{F} \left( T_1',T_2' \right) \leq \varepsilon$, if the distances between starting and ending points of $T_1$ and $T_2$ are not greater than $\varepsilon$.
\end{theorem}
By using the fact, that the length of an $xy$-monotone path in $F$ is equal to the sum of the lengths of subpaths going through free and forbidden space we derive the following about the \MaxIn\ problem.
\begin{cor}\label{cor:nonoptimalResult2}
	Given two polygonal curves $T_1$ and $T_2$ in the plane, an arbitrary fixed $\varepsilon \ge 0$ and an approximation parameter $\delta >0$, we can compute in $\mathcal{O}( \frac{n^4}{\delta^2})$ time a pair of parameterizations $\left( \tilde\alpha_1 , \tilde\alpha_2 \right)$ and its quality $\quality{\tilde{\alpha}}{W}$ for $T_1$ and $T_2$, such that
			$\Tquality{T}{W} - \delta \cdot \max \left\lbrace |T_1|, |T_2| \right\rbrace\leq \quality{\tilde{\alpha}}{W} \leq \Tquality{T}{W}$.
Furthermore, we can construct two polygonal curves, $T_1'$ and $T_2'$, realizing $\quality{\tilde{\alpha}}{W}$, such that $\delta_{F} \left( T_1',T_2' \right) \leq \varepsilon$, if the distances between starting and ending points of $T_1$ and $T_2$ are not greater than $\varepsilon$.
\end{cor}

\begin{proof}

Let $\widetilde{\pi}_{\sPoint \tPoint} \subset F$ be the path reported by our algorithm between $\sPoint$ and $\tPoint$ for \ShortMP\ problem, and $\pi_{\sPoint \tPoint}$ be the true shortest weighted $xy$-monotone path.
Assume $\pi_{\sPoint \tPoint}^B$ ($\pi_{\sPoint \tPoint}^W$) is the sum of the unweighted length of subpaths of $\pi_{\sPoint \tPoint}$, going through the forbidden (free) space. Analogously we define $\widetilde{\pi}_{\sPoint \tPoint}^B$ and $\widetilde{\pi}_{\sPoint \tPoint}^W$ for the path $\widetilde{\pi}_{\sPoint \tPoint}$.
Since the optimum for \LongMP\ has the longest subpath in the free space, it follows that the quality of the solution for \MaxIn\ problem (provided by $\widetilde{\pi}_{\sPoint \tPoint}$) is $\quality{\tilde{\alpha}}{W} \leq \Tquality{T}{W}$.
Also following equations hold.

\begin{align}
\pi_{\sPoint \tPoint}^B + \pi_{\sPoint \tPoint}^W = |T_1| + |T_2| &\Leftrightarrow \pi_{\sPoint \tPoint}^B = |T_1| + |T_2| - \pi_{\sPoint \tPoint}^W\\
\widetilde{\pi}_{\sPoint \tPoint}^B + \widetilde{\pi}_{\sPoint \tPoint}^W = |T_1| + |T_2| &\Leftrightarrow \widetilde{\pi}_{\sPoint \tPoint}^B = |T_1| + |T_2| - \widetilde{\pi}_{\sPoint \tPoint}^W
\end{align}

where $|T_i|$ is the length of $T_i$, for $i = 1,2$. Lemma \ref{lem:appPathG} gives us $\widetilde{\pi}_{\sPoint \tPoint}^B \leq \pi_{\sPoint \tPoint}^B + \delta \cdot \max \{ |T_1|,|T_2| \}$.
\begin{alignat*}{4}
		 &\widetilde{\pi}_{\sPoint \tPoint}^B & \leq & \pi_{\sPoint \tPoint}^B + \delta \cdot \max \{ |T_1|,|T_2| \}\\
		 \stackrel{(10)}{\Rightarrow}& |T_1| + |T_2| - \widetilde{\pi}_{\sPoint \tPoint}^W & \leq & \pi_{\sPoint \tPoint}^B + \delta \cdot \max \{ |T_1|,|T_2| \}\\
		 \stackrel{(9)}{\Rightarrow}& |T_1| + |T_2| - \widetilde{\pi}_{\sPoint \tPoint}^W & \leq & |T_1| + |T_2| - \pi_{\sPoint \tPoint}^W + \delta \cdot \max \{ |T_1|,|T_2| \}\\
		 \Rightarrow & \pi_{\sPoint \tPoint}^W - \delta \cdot \max \{ |T_1|,|T_2| \} & \leq &  \widetilde{\pi}_{\sPoint \tPoint}^W.
\end{alignat*}
Since in the context of \MaxIn\ problem $\quality{\tilde{\alpha}}{W} = \widetilde{\pi}_{\sPoint \tPoint}^W$ and $ \Tquality{T}{W} = \pi_{\sPoint \tPoint}^W$, we get the claimed approximation accuracy. 

\end{proof}
\algrenewcommand{\algorithmicrequire}{\textbf{Input:}}
\algrenewcommand{\algorithmicensure}{\textbf{Output:}}
\section{Improvement}
\label{sec:improvement}
First, we present an abstract problem and then use it to improve the solutions to the  \MinEx\ and \MaxIn\ problems.  
Suppose $\Omega$ is a convex region and $P$ is a set of $n$ points on the boundary of $\Omega$, denoted $\partial \Omega$. 
Two points $p_1$ and $p_2$ are said to be $xy$-monotone if $p_1p_2$ is $xy$-monotone, in which case we write $p_1 \uparrow p_2$.
Our goal is to construct a directed graph $G'(P)$, where the vertices of $G'(P)$ are Steiner points together with the points in $P$ satisfying the following.
For each pair of points $p_i, p_j \in P$ such that  $p_i \uparrow p_j$, there exists a directed path in 
$G'$ between the vertex corresponding to $p_i$ to the vertex corresponding to $p_j$ and 
the cost of this path is $|p_ip_j|_1$, where $|p_ip_j|_1$ is the length of $p_ip_j$ in $L_1$-metric.
We can accomplish this task by examining $\binom{n}{2}$ pairs of points  from $P$.
However, we show that $G'(P)$ can be constructed with the aid of Steiner points
in $\mathcal{O}(n\log n)$ time and size. Our method is based on the following simple geometric observations \cite{Lee90}.

\begin{observation}
Let $a$ (respectively, $b$)  be a point in the South-West (respectively, North-East) quadrant of the Cartesian coordinate system. Then there exists an $L_1$-shortest path from $a$ to $b$ that passes through the origin.
\end{observation}

\begin{observation}\label{obs:L1path}
Let $a$ and $b$ be two points such that $a\uparrow b$.
Any $xy$-monotone path from $a$ to $b$ lies inside the bounding box of $a$ and $b$. Furthermore, all $xy$-monotone paths from $a$ to $b$, have the same length in $L_1$-metric. 
\end{observation}

We compute $G'(P)$ as follows.
%
Initialize  the vertex set of $G'(P)$ to be  $P$. Then add the following Steiner points.
Compute the vertical median line $m$, splitting $P$ into at least $\lfloor \frac{|P|}{2} -1 \rfloor$ points to the left and at least $\lfloor \frac{|P|}{2} -1 \rfloor$ points to the right of $m$, respectively (Figure \ref{fig:generalMethod}a). 
Let $m_1$ be the upper and $m_2$ be the lower intersection points of $m \cap \partial \Omega$. 
Denote by $\ell_1$ and $\ell_2$ the horizontal lines containing $m_1$  and $m_2$, respectively.  
Add $m_1$ and $m_2$ as Steiner points  to the set of vertices of $G'(P)$. 
Partition  $P$ into three sets as follows. Let $P_{\textit{above}}$ be the set of points lying above $\ell_1$, $P_{\textit{below}}$ be the set of points lying below $\ell_2$ and let $P_{\textit{middle}}=P\setminus (P_{\textit{below}} \cup P_{\textit{above}})$.
For all the points $p_i \in P_{\textit{middle}}$, compute their orthogonal projection $p_i^m$ onto $m$. Add $p_i^m$ as a Steiner point. Moreover, add the edge $(p_i, p_i^m)$ to $G'(P)$, directed with respect to $xy$-monotone order. This projection onto $m$ implies an ordering among the points in $P_{\textit{middle}}$, with respect to increasing $y$-coordinate. Let $\langle p_1,...,p_k \rangle = P_{\textit{middle}}$  be this ordering.  For $i = 1,...,k-1$, add  the edges $\left( p_i^m,p_{i+1}^m \right)$ to $G'(P)$,  directed with respect to $xy$-monotone order. 
Also add the edges $\left( m_2,p_1^m \right)$ and $\left( p_k^m,m_1 \right)$. For all vertices $p_i \in P_{\textit{above}}$ (respectively, $p_j \in P_{\textit{below}}$), add the edge $\left( m_1,p_i \right)$ (respectively, $\left( p_j,m_2 \right)$). The weight of each edge is its length in $L_1$-metric. Let $P_{\textit{left}}$ (respectively, $P_{\textit{right}}$) be the set of elements of $P$ lying to the left (respectively, right) of $m$. Recursively apply the construction for $P_{\textit{left}}$ and $P_{\textit{right}}$. 

\begin{figure} [h]
	\centering
\begin{tabular}{ccc}
			\includegraphics[scale=.8]{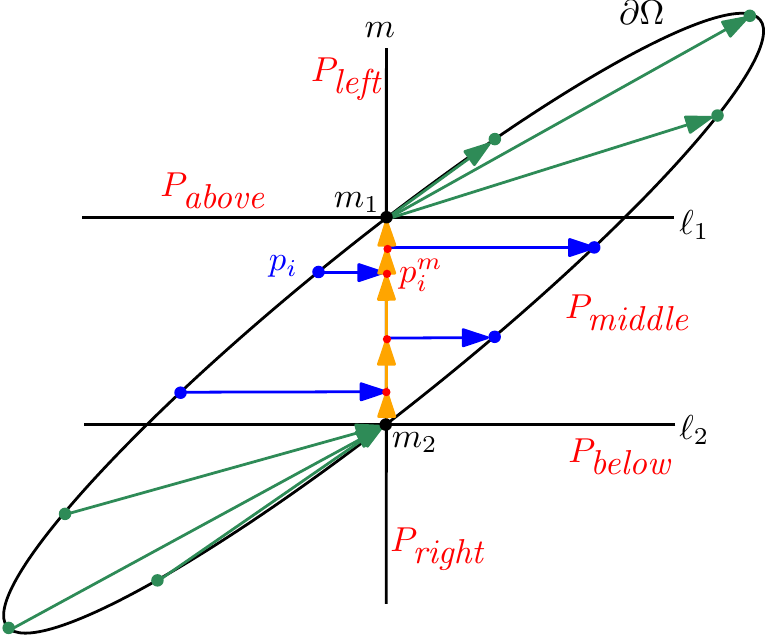}&&\includegraphics[scale=.8]{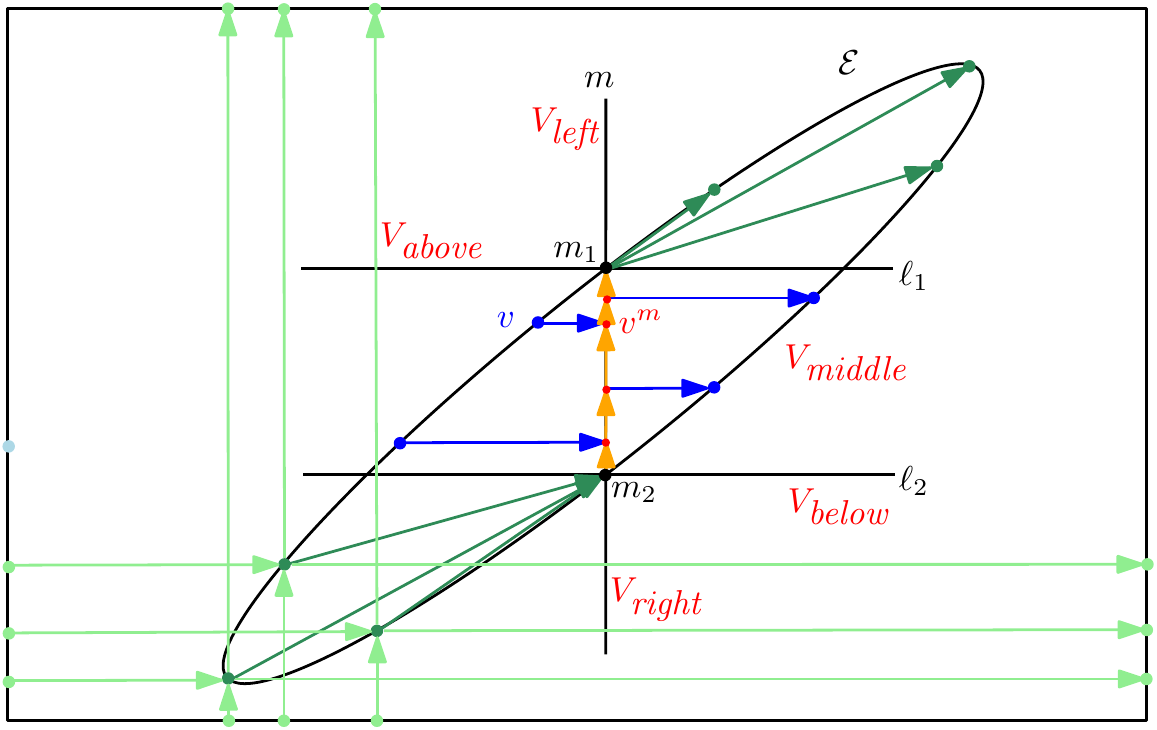}\\
			{(a)}&&{(b)}\\
		\end{tabular}
	\caption{(a) The point set $P$ is partitioned with respect to its median line $m$. Each $p_i \in P_{\textit{middle}}$ is projected onto $m$ (blue arrows and red disks). The projections are ordered with respect to their $y$-coordinates (orange arrows). Each $p_j \in P_{\textit{above}}$ (respectively, $p_j \in P_{\textit{below}}$) is connected to $m_1$ (respectively, $m_2$) (dark green arrows). (b) Each $v\in V_{\partial \mathcal{E}}$ is connected by directed $xy$-monotone edges (light green edges) to all four sides of $\partial C^{i,j}$.}
	\label{fig:generalMethod}
\end{figure}

\newcounter{improvment}
\setcounter{improvment}{\value{lemcounter}}
\begin{lemma}\label{lem:convexConnectivity}
	Let $\Omega$ be a convex region and $P$ be a set of $n$ points on its boundary. For each pair of points $p_i, p_j \in P$ such that $p_i \uparrow p_j$, there exists a path $\widetilde{\pi}_{p_i,p_j}$ in $G'(P)$ between the vertex corresponding to $p_i$ to the vertex corresponding to $p_j$, and its cost is $|p_ip_j|_1$. 
Furthermore, the complexity of $G'(P)$ is $\mathcal{O}(n\log{n})$.
\end{lemma}
\begin{proof}
The points $p_i$ and $p_j$  are separated by a median line in one of the recursive calls. 
Since $p_i\uparrow p_j$, the projection of $p_i$ is below the projection of $p_j$ on that median line. 
Thus, they are connected by an $xy$-monotone path $\widetilde{\pi}_{p_i,p_j}$ in  $G'$.  The cost of this path follows from Observation \ref{obs:L1path}.
The recursion depth is at most $\mathcal{O} \left( \log{n} \right)$ as the problem is partitioned with respect to the median. 
In each call, a linear number of edges and Steiner points are added with respect to the size of the input. Thus $G'$ has $\mathcal{O} \left( n\log{n}\right)$ vertices and directed edges.
\end{proof}
	
\subsection{Improvement for the \MinEx\ and \MaxIn\ problems}
In Section \ref{sec:firstAlgorithm}, we showed how to find an approximate solution to \ShortMP\ in $\mathcal{O} \left(\frac {n^4}{\delta^2} \right)$ time. To do so, we defined a neighborhood graph $G=(V,E)$  of size $\mathcal{O} \left(\frac{n^4}{\delta^2} \right)$, computed a shortest path $\widetilde{\pi}$ in $G$ and proved that 
$\widetilde{\pi}$ approximates the solution to \ShortMP.  In this section,  we show how to compute an approximate solution from a  neighborhood graph $G^*$ of size $\mathcal{O} \left( \frac{n^3}{\delta} \log \left(\frac{n}{\delta} \right) \right)$. 
The graph $G^*$ is defined as follows. For each parameter cell $C^{i,j}$ with ellipse $\mathcal{E}$, we restrict $V$ to the boundaries of $C^{i,j}$ and $\mathcal{E}$. Formally,  
let $V_{\partial C^{i,j}}= \partial C^{i,j}\cap V$ and $V_{{\partial \mathcal{E}}}= \partial \mathcal{E}\cap V$.
The vertex set of $G^*$ for this cell is  $V_{\partial C^{i,j}} \cup V_{{\partial \mathcal{E}}}\cup V'$, where $V'$ is the set of vertices of 
$G'(V_{{\partial \mathcal{E}}})$ as defined above.
\newcounter{vertices-On-Boundary}
\setcounter{vertices-On-Boundary}{\value{lemcounter}}
\begin{lemma}\label{lem:verticesOnBoundary}
	The  size of $V_{\partial C^{i,j}} \cup V_{{\partial \mathcal{E}}}\cup V'$ is  $\mathcal{O} \left( \frac{n}{\delta} \log \left( \frac{n}{\delta} \right) \right)$.
\end{lemma}
\begin{proof}
Based on the construction of $G$, in Section \ref{sec:firstAlgorithm}, there are $\frac{n}{\delta}$ equidistant vertical and horizontal grid lines. In the worst case all of them are intersecting $\partial C^{i,j}$ and $\partial \mathcal{E}$.
Also for each intersection point on $\partial \mathcal{E}$ there exists an intersection line. So, in the worst case, the total number of vertices of $G$ on $\partial C^{i,j}$ and $\partial \mathcal{E}$ (i.e., $V_{\partial C^{i,j}} \cup V_{{\partial \mathcal{E}}}$) is $\frac{4n}{\delta}=\mathcal{O}\left( \frac{n}{\delta} \right)$.
Moreover, since $V'$ is the set of vertices of $G'(V_{{\partial \mathcal{E}}})$, based on Lemma \ref{lem:convexConnectivity} the complexity of $G'(V_{{\partial \mathcal{E}}})$ is $\mathcal{O} \left( \frac{n}{\delta} \log \left( \frac{n}{\delta} \right) \right)$.
So, the  size of $V_{\partial C^{i,j}} \cup V_{{\partial \mathcal{E}}}\cup V'$ is  $\mathcal{O}\left( \frac{n}{\delta} \right)+\mathcal{O} \left( \frac{n}{\delta} \log \left( \frac{n}{\delta} \right) \right)=\mathcal{O} \left( \frac{n}{\delta} \log \left( \frac{n}{\delta} \right) \right)$.
\end{proof}
For each parameter cell $C^{i,j}$ with ellipse $\mathcal{E}$, we define the set of edges of $G^*$ as follows.
For each  $v\in V_{\partial \mathcal{E}}$, there is an edge between $v$ and each of its projections on $\partial C^{i,j}$. 
Each of these edges are directed with respect to the $xy$-monotone order and weighted by the length of their intersection with the forbidden space (Figure \ref{fig:generalMethod}b). Moreover, each edge in $G'(V_{\partial\mathcal{E}})$  is an edge in this cell with weight zero.
Finally, $G^*$ is the union of all the graphs defined for  each cell.
\newcounter{GPrime}
\setcounter{GPrime}{\value{lemcounter}}
\begin{lemma}\label{lem:GPrimeCompl}
	The size of  $G^*$ is $\mathcal{O} \left( \frac{n^3}{\delta}    \log \left(\frac{n}{\delta} \right) \right)$.
\end{lemma}

\noindent We now explain how to compute an approximate solution to \ShortMP\ by using $G^*$ as a neighborhood graph.
\begin{lemma}\label{lem:GPrimeShorterThanG}
	Let $\widetilde{\pi}_{s't'}\subset G$ be a weighted shortest path connecting two points on the boundary of $C^{i,j}$. Then, there exists a path $\widetilde{\pi}_{s't'}^*$ in $G^*$, such that $||\widetilde{\pi}_{s't'}^*|| \leq ||\widetilde{\pi}_{s't'}||$.
\end{lemma}
\begin{proof}
	Let $\mathcal{E}$ be the ellipse corresponding to $C^{i,j}$. Since $\widetilde{\pi}_{s't'}$ corresponds to an $xy$-monotone path, there exists at most one vertex $a$ (respectively, $b$) where $\widetilde{\pi}_{st}$ enters (respectively, exits) $\mathcal{E}$. 
So, $\widetilde{\pi}_{s't'}$ consists of three sub-paths $\widetilde{\pi}_{s'a} \subset B$, $\widetilde{\pi}_{ab} \subset W$ and $\widetilde{\pi}_{bt'} \subset B$, whose concatenation is $\widetilde{\pi}_{s't'}$. Here, $a$ is defined as the last (respectively, first) vertex of $\widetilde{\pi}_{s'a}$ (respectively, $\widetilde{\pi}_{ab}$) and $b$ as the last (respectively, first) vertex of $\widetilde{\pi}_{ab}$ (respectively, $\widetilde{\pi}_{bt'}$).
Since $\widetilde{\pi}_{s'a}$ is $xy$-monotone, it follows that $a$ dominates $s'$. Moreover,  $a$ dominates the orthogonal projection of $a$ onto  the grid line on which $s'$ lies. This implies that there exists a path $\widetilde{\pi}_{s'a}^*$ connecting $s'$ and $a$. Since $\widetilde{\pi}_{s'a} \subset B$ and both paths have the same ending points, it follows $||\widetilde{\pi}_{s'a}^*|| \le ||\widetilde{\pi}_{s'a}||$.
Since $\widetilde{\pi}_{ab}$ is $xy$-monotone, we have $a \uparrow b$. It follows from Lemma \ref{lem:convexConnectivity} that there exists a path $\widetilde{\pi}_{ab}^* \subset \textit{W}$. This implies that $||\widetilde{\pi}_{ab}^*|| = 0$ (because the weight in free-space is zero). Analogously, to the construction of $\widetilde{\pi}_{s'a}^*$ it follows, that there exists a path in $G^*$, such that $||\widetilde{\pi}_{bt'}^*|| \leq ||\widetilde{\pi}_{bt'}||$. The concatenation of the three constructed paths results in the required path.
\end{proof}

The following theorem is an improvement over Theorem \ref{thm:nonoptimalResult}. The proof is the same except that  for each parameter cell,  we apply Lemma \ref{lem:GPrimeShorterThanG} instead of Lemma \ref{lem:appPathG}. 
\begin{theorem}\label{thm:finalResult}
	Given two polygonal curves $T_1$ and $T_2$ in the plane, an arbitrary fixed $\varepsilon > 0$ and an approximation parameter $\delta >0$, we can compute in $\mathcal{O} \left( \frac{n^3}{\delta}   \log \left( \frac{n}{\delta} \right) \right)$ time a pair of parameterizations $\left( \tilde\alpha_1 , \tilde\alpha_2 \right)$ and its quality $\quality{\tilde{\alpha}}{B}$ for $T_1$ and $T_2$, such that
			$\Tquality{T}{B} \leq \quality{\tilde{\alpha}}{B} \leq \Tquality{T}{B} + \delta \cdot \max \left\lbrace |T_1|, |T_2| \right\rbrace$.
Furthermore, we can construct two polygonal curves, $T_1'$ and $T_2'$, realizing $\quality{\tilde{\alpha}}{B}$, such that $\delta_{F} \left( T_1',T_2' \right) \leq \varepsilon$, if the distances between starting and ending points of $T_1$ and $T_2$ are not greater than $\varepsilon$.
\end{theorem}
\begin{cor}\label{cor:finalResult}
Given two polygonal curves $T_1$ and $T_2$ in the plane, an arbitrary fixed $\varepsilon > 0$ and an approximation parameter $\delta >0$, we can compute in $\mathcal{O} \left( \frac{n^3}{\delta}   \log \left( \frac{n}{\delta} \right) \right)$ time a pair of parameterizations $\left( \tilde\alpha_1 , \tilde\alpha_2 \right)$ and its quality $\quality{\tilde{\alpha}}{W}$ for $T_1$ and $T_2$, such that
			$\Tquality{T}{W} - \delta \cdot \max \left\lbrace |T_1|, |T_2| \right\rbrace\leq \quality{\tilde{\alpha}}{W} \leq \Tquality{T}{W}$.
Furthermore, we can construct two polygonal curves, $T_1'$ and $T_2'$, realizing $\quality{\tilde{\alpha}}{W}$, such that $\delta_{F} \left( T_1',T_2' \right) \leq \varepsilon$, if the distances between starting and ending points of $T_1$ and $T_2$ are not greater than $\varepsilon$.
\end{cor}

\subsection{Is FPTAS achievable?}\label{sec:FPTAS}
Har-Peled and Wang \cite{Yusu06} proposed an approximation algorithm for the \MaxIn\ problem whose running time is $\mathcal{O} \left(\frac{n^4}{\delta^2}\right)$, where $n$ is the size of the input polygonal curves.
They claimed that their approach is an $\left( 1 - \delta \right)$-approximation, that is, the quality of their solution is greater than $\left( 1 - \delta \right)$ times that of the optimal (\cite{Yusu06}, Theorem 4.2).
Their proof is based on a claim that the length of the boundary of the free-space of an arbitrary parameter cell is bounded by $\mathcal{O} \left( \Tquality{T}{W} \right)$.
However, we show via a counterexample (see Appendix) that this claim is not true (see also \cite{Yusu13}).

It remains to be discussed if it is possible to design a FPTAS (Fully Polynomial-Time Approximation Scheme) for the $\MaxIn$ problem. 
As already discussed, asking for a polynomial time $\left( 1 - \delta \right)$-approximation for the $\MaxIn$ problem is equivalent to searching for a polynomial time $\left( 1 - \delta \right)$-approximation algorithm for the $\LongMP$ problem. 
The common methodology to approximate longest paths in the presence of weighted regions is to approximate those areas by ``accurate enough'' structures that are more easily manageable, than the original ones.
One frequently used approach is to overlay the relevant weighted areas, in our case the connected components of $W$,  by a mesh, having a (small enough) grid size, related to $\delta$. 
For example, this could be achieved by carefully connecting additionally positioned Steiner points on $W$ to a dense enough neighborhood graph.
To guarantee a $\left( 1 - \delta \right)$-approximation for the $\MaxIn$ problem we have to upper bound the resulting error, i.e. the grid size, by $\delta$ times the optimal solution, that is by $\delta \cdot \Tquality{T}{W}$. As illustrated in Figures \ref{fig:counterExample}(a) and (b) shown in the Appendix, we can imagine configurations, in which the size of the areas to be covered could be arbitrarily larger than the quality provided by an optimal solution.
This implies directly an arbitrary large complexity of the applied grid-like structure.
It is easy to see, that we can make the analogous observations for the $\MinEx$ problem, that is, the quality of the optimal solution $\Tquality{T}{B}$ could be arbitrary small (see Figure \ref{fig:outline} for an illustration).
One typical way which leads to efficient approximation algorithms is to argue approximating objects by applying some ``fatness" properties \cite{Aleks10,Berg10,Stappen93}. Unfortunately, there is no indication for such properties being fulfilled in general in the context of our problem settings.

\begin{figure}[ht]
    		\begin{center}
     			 \begin{tabular}{ccc}
       				\includegraphics[height=2.9cm]{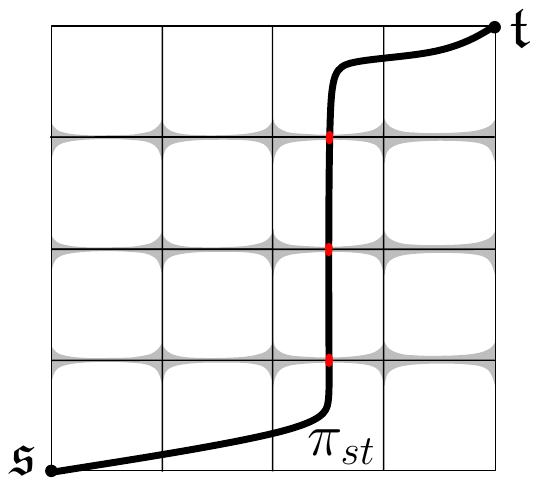} & \includegraphics[height=2.9cm]{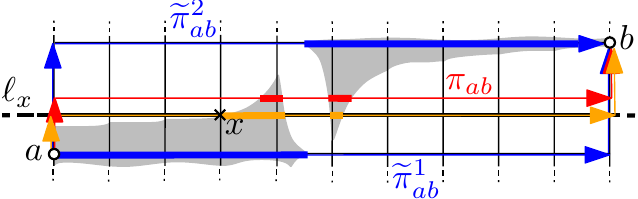}\\
       				{\small a) An arbitrarily small $\Tquality{T}{B}$}.  & {\small b)  A path decision could be locally good, but maybe globally bad.} \\
      			\end{tabular}
    		\end{center}
    		\vspace*{-12pt}
    		\caption{a) A free-space diagram in which the weighted parts (red subcurves) of a weighted shortest $xy$-monotone path $\pi_{\sPoint \tPoint}$ are arbitrary small, compared to the length of $T_1$ and $T_2$. b) A weighted shortest subpath $\pi_{ab}$ and the two possibilities $\widetilde{\pi}_{ab}^1$ and $\widetilde{\pi}_{ab}^2$, to stay on $G$ close to $\pi_{ab}$. Both solutions have a weight (sum of the lengths of the fat blue subpaths) much bigger than that of an optimal (sum of the lengths of the fat red subpaths). An additional intersection line, determined by an intersection point $x$ of the boundary of an ellipse and a grid line, enables a path $\widetilde{\pi}_{\sPoint \tPoint}$, whose weighted subpaths are much smaller (sum of the lengths of the fat orange subpaths).}
   		\label{fig:outline}
\end{figure}


Although it seems difficult to achieve a FPTAS, it is possible to design a polynomial time $(1-\delta)$-approximation algorithm that depends on factors other than the size of the input and the approximation factor $\delta$.
Based on Observation \ref{obs:transformation}, any solution for \LongMP\ is transformable to a solution for \MaxIn\ problem. Therefore, an approximation algorithm for \LongMP\ is an approximation algorithm for \MaxIn\ as well.
These problems are solvable exactly in the $L_1$-metric. In the following algorithm, we make use of this to design a  $\delta$-approximation algorithm in $L_2$-metric.

Let $F_1$ (respectively $F_2$) be the the Free Space diagram for the $L_1$ (respectively $L_2$)-metric with Fr\'{e}chet distance $\eps$.\\
{\bf Step 1:} Compute an exact solution for \LongMP\ problem in $L_1$-metric using \cite{Yusu09}. Let the length of this optimal solution be $\gamma$.\\
{\bf Step 2:} For each parameter cell of $F_2$, $C^{i,j}$, with free-space $\mathcal{E}$, position equally spaced Steiner points on $\partial\mathcal{E}$ with a distance of $\frac{\gamma\delta}{4n}$, where $n=n_1+n_2$.\\
{\bf Step 3:} Construct a directed acyclic graph (DAG) on these Steiner points as in the algorithm in Corollary \ref{cor:finalResult}, and find the longest path $\widetilde{\pi}_{\sPoint\tPoint}$ from $\sPoint$ to $\tPoint$ in this DAG.\\
Let $\pi_{\sPoint\tPoint}$ be an optimum solution for \LongMP\ in $L_2$-metric. We claim the following lemma.
\begin{lemma}
\label{lem:delta}
$\gamma \leq ||\pi_{\sPoint\tPoint}||$.
\end{lemma}
\begin{proof}
Suppose $W_{1}$ (respectively $W_{2}$) is the free-space of $F_1$ (respectively $F_2$).
For any pair of matched points $(a,b)\in W_{1}$ we know that $|ab|_1 \leq \eps$, where $|.|_1$ is $L_1$ distance.
Also, it is known that $|ab| \leq |ab|_1$, where $|.|$ is the Euclidean distance.
Therefore, $|ab| \leq \eps$ implies $(a,b) \in W_{2}$. This implies that $W_{1} \subseteq W_{2}$ . Therefore, the longest path in $W_{1}$ is shorter than or equal to the longest path in $ W_{2}$.
\end{proof}

\begin{theorem}
Given two polygonal curves $T_1$ and $T_2$ in the plane, an arbitrary fixed $\varepsilon > 0$ and an approximation parameter $\delta >0$, we can compute in $\mathcal{O} \left( \frac{n^3}{\gamma\delta} \log \left( \frac{n}{\gamma\delta} \right) \right)$ time a pair of parameterizations $\left( \tilde\alpha_1 , \tilde\alpha_2 \right)$ and its quality $\quality{\tilde{\alpha}}{W}$ for $T_1$ and $T_2$, such that
			$(1-\delta) \cdot \Tquality{T}{W} \leq \quality{\tilde{\alpha}}{W} \leq \Tquality{T}{W}$,
where $\gamma$ is the optimal solution in $L_1$-metric.
Furthermore, we can construct two polygonal curves, $T_1'$ and $T_2'$, realizing $\quality{\tilde{\alpha}}{W}$, such that $\delta_{F} \left( T_1',T_2' \right) \leq \varepsilon$, if the distances between starting and ending points of $T_1$ and $T_2$ are not greater than $\varepsilon$.
\end{theorem}

\begin{proof}
Since the solution generated by our algorithm $\widetilde{\pi}_{\sPoint\tPoint}$ is $xy$-monotone it intersects no more than $n$ cells. The maximum error in each cell is at most $4*\frac{\gamma\delta}{4n}$. 
The reason is, as illustrated in Figure \ref{fig:deltaApp}, the difference between the length of the optimum path $\pi_{\sPoint\tPoint}$ in the free-space of a cell and the path between its adjacent Steiner points in the graph (between two of $s_1$, $s_2$, $s'_1$ and $s'_2$ in the figure) is at most 4 times the distance between two consecutive Steiner points. 
Recall that we measure the length of the paths in free-space diagram with $L_1$-metric (dotted lines in Figure \ref{fig:deltaApp}).
Hence $||\widetilde{\pi}_{\sPoint\tPoint}|| \ge ||\pi_{\sPoint\tPoint}||- n*4*\frac{\gamma\delta}{4n}$. By Lemma \ref{lem:delta} it follows that $||\widetilde{\pi}_{\sPoint\tPoint}|| \ge (1-\delta)||\pi_{\sPoint\tPoint}||$.
The running time of the algorithm involves computing  the exact solution for \LongMP\ problem using the algorithm of \cite{Yusu09}, placement of Steiner points,  construction of the DAG and finding a path in this DAG.  
Putting everything together, the time complexity is $\mathcal{O} \left( \frac{n^3}{\gamma\delta} \log \left( \frac{n}{\gamma\delta} \right) \right)$.
\end{proof}
\begin{figure}[ht]
	\begin{center}
       			\includegraphics{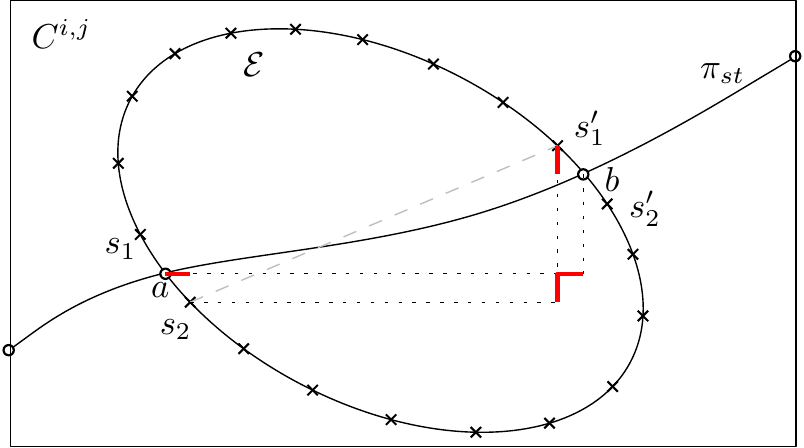} 
	\end{center}
    		\caption{The length of the path inside the free-space (which is convex) is shown by dotted lines. The difference between the length of the optimum path $\pi_{\sPoint\tPoint}$ in free-space of a cell and the path between its adjacent Steiner points $s_2$ and $s'_1$ in the graph is at most 4 times the length of red line segments.}
   		\label{fig:deltaApp}
\end{figure}
\section{Conclusion}\label{sec:conclusion}
Our approach to measure the similarity between two polygonal curves in the presence of outliers is to minimize the portions of the curves where the matching is not possible (\MinEx\ problem) or maximize the matched subcurves (\MaxIn\ problem).
We reduced our problems to that of finding an $xy$-monotone weighted path in the deformed free-space diagram in $L_1$-metric. For the \MinEx\ problem (\MaxIn\ problem) the free-space is weighted by zero (one) and the forbidden-space is weighted by one (zero). 
After proving that these problems are not solvable in the \acmq, we designed approximation algorithms for the \MinEx\ and \MaxIn\ problems.
We proposed an algorithm, running in $\mathcal{O} \left( \frac{n^3}{\delta} \log \left( \frac{n}{\delta} \right)\right)$ time with additive approximation error.
It is still open if there exist a FPTAS for this problem. However, as we have shown, it is possible to design a $(1-\delta)$-approximation algorithm that its complexity depends on the input size $n$, given approximation factor $\delta$ and the length of the optimum solution in $L_1$-metric $\gamma$.

\section{Acknowledgment}
The authors thank Yusu Wang for communication \cite{Yusu13} regarding clarifications on Theorem 4.2 \cite{Yusu06}.




\bibliography{references}{}

\break
\section{Appendix}

Consider the following example, where each trajectory $T_1$ and $T_2$ consists of a single line segment (Figure \ref{fig:counterExample}a). They have the same length and are parallel to each other at distance $\eps$.
Their starting and ending points are in opposite directions. Therefore, the free-space for $T_1$ and $T_2$ at Fr\'{e}chet distance $\eps$ is equal to the diagonal, connecting the upper left corner of $F$ to the bottom right corner. 
Suppose that $\omega$ is the length of the boundary of free-space of the parameter cell.
By moving $T_1$ toward $T_2$, the free space becomes a 2-dimensional solid. This solid can be chosen arbitrarily thin, such that the ratio $\frac{\omega}{\Tquality{T}{W}}$ is arbitrarily big. This is a contradiction to $\omega = \mathcal{O} \left( \Tquality{T}{W} \right)$, which was an assumption in Theorem 4.2 \cite{Yusu06}.\\


\begin{figure}[ht]
    		\begin{center}
     			 \begin{tabular}{ccc}
       				\includegraphics[height=2.9cm]{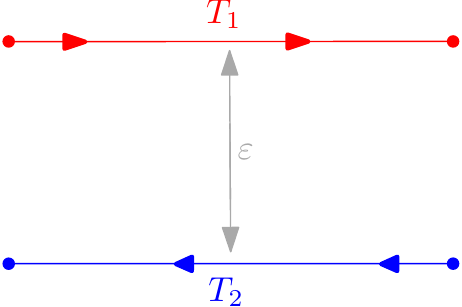} & \includegraphics[height=2.9cm]{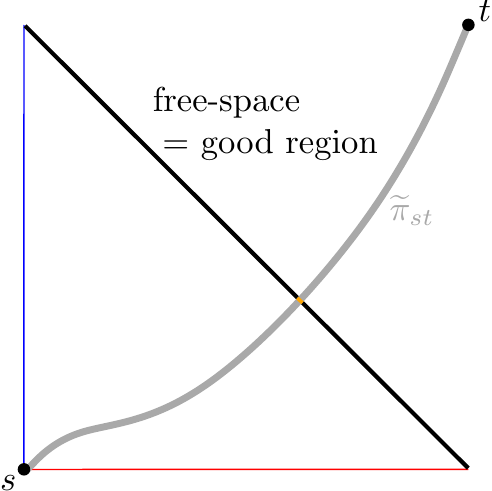}\\
       				{\small a) Two segments lie parallel to each other }  & {\small b) Free-space of $T_1$ and $T_2$ is equal to} \\
				{\small \;\;\;\;\;\;  with the distance equal to $\eps$. They have } & {\small \;\;\;\; the diagonal, connecting the upper}\\
				{\small \;\;\;\;\;\;\;\;\;opposite direction for movement. } & {\small \;\;\;\;\;\;\;\;\; left corner with the bottom right one.}\\

      			\end{tabular}
    		\end{center}
    		\vspace*{-12pt}
    		\caption{Counterexample for $\omega = \mathcal{O} \left( \Tquality{T}{W} \right)$}
   		\label{fig:counterExample}
\end{figure}

\end{document}